\documentclass{elsart}

\usepackage{graphicx}
\usepackage{amsmath}
\usepackage{amssymb}
\usepackage{amsthm}
\usepackage{mathtools}
\usepackage{color}

\usepackage{algorithm}
\usepackage{algpseudocode}

\usepackage{todonotes}
\usepackage{xspace,xargs}
\usepackage{url}
\usepackage{multirow}

\newcommand{\repo}{\url{https://github.com/nicolaprezza/bwt2lcp}\xspace}
\newcommand{\eGap}{\texttt{eGap}\xspace}

\def\mergeBwtLCP{\texttt{merge}\xspace}
\def\induceLCP{\texttt{bwt2lcp}\xspace}

\usepackage{xspace,xargs}
\def\bigS{\mathcal{S}\xspace}

\newtheorem{theorem}{\ \\Theorem}
\newtheorem{lemma}{\ \\Lemma}
\theoremstyle{definition}
\newtheorem{definition}{\ \\Definition}

\begin{document}

\begin{frontmatter}

\thanks[SIR]{ GR is partially, and NP is totally, supported by the project MIUR-SIR CMACBioSeq (``Combinatorial methods for analysis and compression of biological sequences'') grant n.~RBSI146R5L.}

\title{Space-Efficient Construction of Compressed Suffix Trees}

\author{
Nicola Prezza\thanksref{SIR}}

\address{
Department of Computer Science, University of Pisa, Italy
{\tt nicola.prezza@di.unipi.it}}

\author{Giovanna Rosone\thanksref{SIR}}

\address{
Department of Computer Science, University of Pisa, Italy
{\tt giovanna.rosone@unipi.it}}

\begin{abstract}

We show how to build several data structures of central importance to string processing, taking as input the Burrows-Wheeler transform (BWT) and using small extra working space. Let $n$ be the text length and $\sigma$ be the alphabet size. 
We first provide two algorithms that enumerate all LCP values and suffix tree intervals in $O(n\log\sigma)$ time using just $o(n\log\sigma)$ bits of working space on top of the input BWT. 
Using these algorithms as building blocks, for any parameter $0 < \epsilon \leq 1$ we show how to build the PLCP bitvector and the balanced parentheses representation of the suffix tree topology in $O\left(n(\log\sigma + \epsilon^{-1}\cdot \log\log n)\right)$ time using at most $n\log\sigma \cdot(\epsilon + o(1))$ bits of working space on top of the input BWT and the output. 
In particular, this implies that we can build a compressed suffix tree from the BWT using just succinct working space (i.e. $o(n\log\sigma)$ bits) and any time in $\Theta(n\log\sigma) + \omega(n\log\log n)$.
This improves the previous most space-efficient algorithms, which worked in $O(n)$ bits and $O(n\log n)$ time. 
We also consider the problem of merging BWTs of string collections, and provide a solution running in $O(n\log\sigma)$ time and using just $o(n\log\sigma)$ bits of working space. 
An efficient implementation of our LCP construction
and BWT merge algorithms 
use (in RAM) as few as $n$ bits on top of a packed representation of the input/output and process data as fast as $2.92$ megabases per second.
\begin{keyword}
Burrows-Wheeler transform, compressed suffix tree, LCP, PLCP.
\end{keyword}

\end{abstract}

\end{frontmatter}

\section{Introduction and Related Work}
The increasingly-growing production of large string collections---especially in domains such as biology, where new generation sequencing technologies can nowadays generate Gigabytes of data in few hours---is lately generating much interest towards fast and space-efficient algorithms able to index this data. The Burrows-Wheeler Transform~\cite{burrows1994block} and its extension to sets of strings~\cite{MantaciRRS07,BauerCoxRosoneTCS2013} is becoming the gold-standard in the field: even when not compressed, its size is one order of magnitude smaller than classic suffix arrays (while preserving many of their indexing capabilities). 
This generated considerable interest towards fast and space-efficient BWT construction algorithms~\cite{BauerCoxRosoneTCS2013,Karkkainen:2007:FBS:1314704.1314852,10.1007/978-3-319-15579-1_46,10.1007/978-3-319-02432-5_5,8712716,Kempa:2019:SSS:3313276.3316368,navarro2014optimal,10.1007/978-3-319-15579-1_46}. 
As a result, the problem of building the BWT is well understood to date. The fastest algorithm solving this problem operates in sublinear $O(n/\sqrt{\log n})$ time and $O(n)$ bits of space on a binary text of length $n$ by exploiting word parallelism~\cite{Kempa:2019:SSS:3313276.3316368}. The authors also provide a conditional lower bound suggesting that this running time might be optimal. The most space-efficient algorithm terminates in $O(n\log n/\log\log n)$ time and uses just $o(n\log\sigma)$ bits of space (succinct) on top of the input and output~\cite{navarro2014optimal}, where $\sigma$ is the alphabet's size. In the average case, this running time can be improved to $O(n)$ on constant-sized alphabets while still operating within succinct space~\cite{10.1007/978-3-319-15579-1_46}.

In some cases, a BWT alone is not sufficient to complete efficiently particular string-processing tasks. 
For this reason, the functionalities of the BWT are often extended by flanking to it additional structures such as the Longest Common Prefix (LCP) array~\cite{CGRS_JDA_2016} (see e.g.~\cite{prezza2018detecting,prezza2019,GuerriniRosoneAlcob2019,TUSTUMI2016} for bioinformatic applications requiring this additional component). 
A disadvantage of the LCP array is that it requires $O(n\log n)$ bits to be stored in plain form. To alleviate this problem, usually the PLCP array~\cite{Sadakane:2002:SRL:545381.545410}---an easier-to-compress permutation of the LCP array---is preferred. The PLCP relies on the idea of storing LCP values in text order instead of suffix array order. As shown by Kasai et al.~\cite{10.1007/3-540-48194-X_17}, this permutation is almost increasing ($PLCP[i+1] \geq PLCP[i]-1$) and can thus be represented in just $2n$ bits in a bitvector known as the \emph{PLCP bitvector}. More advanced applications might even require full suffix tree functionality. In such cases, compressed suffix trees~\cite{Grossi:2005:CSA:1093654.1096192,Sadakane:2007:CST:1326296.1326297} (CSTs) are the preferred choice when the space is at a premium. A typical compressed suffix tree is formed by a compressed suffix array (CSA), the PLCP bitvector, and a succinct representation of the suffix tree topology~\cite{Sadakane:2007:CST:1326296.1326297} (there exist other designs, see Ohlebusch et al.~\cite{10.1007/978-3-642-16321-0_34} for an exhaustive survey).
To date, several practical algorithms have been developed to solve the task of building \emph{de novo} such additional components~\cite{CGRS_JDA_2016,holt2014constructing,holt2014merging,Bonizzoni2018,EgidiAMB2019,Belazzougui:2014:LTC:2591796.2591885,10.1007/978-3-642-02441-2_17,Valimaki:2010:ECS:1498698.1594228}, but little work has been devoted to the task of computing them from the BWT in little working space (internal and external). 
Considering the advanced point reached by state-of-the-art BWT construction algorithms, it is worth to explore whether such structures can be built more efficiently starting from the BWT, rather than from the raw input text.

\textbf{CSA} As far as the CSA is concerned, this component can be easily built from the BWT using small space as  it is formed (in its simplest design) by just a BWT with rank/select functionality enhanced with a suffix array sampling, see also~\cite{Belazzougui:2014:LTC:2591796.2591885}.

\textbf{LCP} We are aware of only one work building the LCP array in small space from the BWT: Beller et al.~\cite{beller2013computing} show how to build the LCP array in $O(n\log\sigma)$ time and $O(n)$ bits of working space on top of the input BWT and the output. 
Other works~\cite{munro2017space,Belazzougui:2014:LTC:2591796.2591885} show how to build the LCP array directly from the text in $O(n)$ time and $O(n\log\sigma)$ bits of space (compact).

\textbf{PLCP} K{\"a}rkk{\"a}inen et al.~\cite{10.1007/978-3-642-02441-2_17} show that the PLCP bitvector can be built in $O(n\log n)$ time using $n$ bits of working space on top of the text, the suffix array, and the output PLCP. Kasai at al.'s lemma also stands at the basis of a more space-efficient algorithm from V\"{a}lim\"{a}ki et al.~\cite{Valimaki:2010:ECS:1498698.1594228}, which computes the PLCP from a CSA in $O(n\log n)$ time using constant working space on top of the CSA and the output. 
Belazzougui~\cite{Belazzougui:2014:LTC:2591796.2591885} recently presented an algorithm for building the PLCP bitvector from the text in optimal $O(n)$ time and compact space ($O(n\log\sigma)$ bits). 

\textbf{Suffix tree topology} The remaining component required to build a compressed suffix tree (in the version described by Sadakane~\cite{Sadakane:2007:CST:1326296.1326297}) is the suffix tree topology, represented either in BPS~\cite{Munro:1997:SRB:795663.796328} (balanced parentheses) or DFUDS~\cite{Benoit2005} (depth first unary degree sequence), using $4n$ bits.
As far as the BPS representation is concerned, Hon et al.~\cite{Hon:2009:BTB:1654348.1654351} show how to build it from a CSA in $O(n(\log\sigma + \log^\epsilon n))$ time and compact space for any constant $\epsilon>0$. Belazzougui~\cite{Belazzougui:2014:LTC:2591796.2591885} improves this running time to the optimal $O(n)$, still working within compact space. 
V\"{a}lim\"{a}ki et al.~\cite{Valimaki:2010:ECS:1498698.1594228} describe a linear-time algorithm that improves the space to $O(n)$ bits on top of the LCP array (which however needs to be represented in plain form), while Ohlebusch et al.~\cite{10.1007/978-3-642-16321-0_34} show how to build the DFUDS representation of the suffix tree topology in $O(t_{lcp}\cdot n)$ 
time using $n+o(n)$ bits of working space on top of a structure supporting access to LCP array values in $O(t_{lcp})$ time.

Summing up, the situation for building compressed suffix trees from the BWT is the following: algorithms working in optimal linear time require $O(n\log\sigma)$ bits of working space. 
Algorithms reducing this space to $O(n)$ (on top of a CSA) are only able to build the suffix tree topology within $O(n\cdot t_{lcp})$ time, which is $\Omega(n\log^{\epsilon}n)$ with the current best techniques, and the PLCP bitvector in $O(n\log n)$ time. No algorithm can build all the three CST components within $o(n\log\sigma)$ bits of working space on top of the input BWT and the output. Combining the most space-efficient existing algorithms, the following two trade-offs can therefore be achieved for building all compressed suffix tree components from the BWT: 
\begin{itemize}
	\item $O(n\log\sigma)$ bits of working space and $O(n)$ time, or
	\item $O(n)$ bits of working space and $O(n\log n)$ time.
\end{itemize}

\paragraph{Our contributions}

In this paper, we give new space-time trade-offs that allow building the CST's components in smaller working space (and in some cases even faster) with respect to the existing solutions. We start by combining Beller et al.'s algorithm~\cite{beller2013computing} with the suffix-tree enumeration procedure of Belazzougui~\cite{Belazzougui:2014:LTC:2591796.2591885} to obtain an algorithm that enumerates (i) all pairs $(i,LCP[i])$, and (ii) all suffix tree intervals in $O(n\log\sigma)$ time using just $o(n\log\sigma)$ bits of working space on top of the input BWT. We use this procedure to obtain algorithms that build (working space is on top of the input BWT and the output):

\begin{enumerate}
    \item\label{goalLCP} The LCP array of a string collection in $O(n\log\sigma)$ time and $o(n\log\sigma)$ bits of working space (see Section \ref{sec:LCP}).
    \item\label{goalPLCP} the PLCP bitvector and the BPS representation of the suffix tree topology in $O\left(n(\log\sigma + \epsilon^{-1}\cdot \log\log n)\right)$ time and $n\log\sigma \cdot (\epsilon + o(1))$ bits of working space, for any user-defined parameter $0 < \epsilon \leq 1$ (see Section \ref{sec:PLCP} and \ref{sec:ST topology}).
    \item\label{goalMerge} The BWT of the union of two string collections of total size $n$ in $O(n\log\sigma)$ time and $o(n\log\sigma)$ bits of working space, given the BWTs of the two collections as input (see Section \ref{sec:algo2}). 
\end{enumerate}

Contribution (\ref{goalLCP}) is the first showing that the LCP array can be induced from the BWT using succinct working space \emph{for any alphabet size}. 

Contribution (\ref{goalPLCP}) can be used to build a compressed suffix tree from the BWT using just $o(n\log\sigma)$ bits of working space and any time in $O(n\log\sigma) + \omega(n\log\log n)$---for example, $O(n(\log\sigma + (\log\log n)^{1+\delta}))$, for any $\delta>0$. On small alphabets, this improves both working space and running time of existing $O(n)$-bits solutions.

Also contribution (\ref{goalMerge})
improves the state-of-the-art, due to Belazzougui et al.~\cite{Belazzougui:2014:LTC:2591796.2591885,belazzougui2016linear}. In those papers, the authors  show how to merge the BWTs of two texts $T_1, T_2$ and obtain the BWT of the collection $\{T_1, T_2\}$ in $O(nk)$ time and $n\log\sigma(1+1/k) + 11n + o(n)$ bits of working space for any $k \geq 1$~\cite[Thm. 7]{belazzougui2016linear}.
When $k=\log\sigma$, this running time is the same as our result (\ref{goalMerge}), but the working space is much higher  on small alphabets.

We implemented and tested our algorithms (\ref{goalLCP}, \ref{goalMerge}) on DNA alphabet. Our tools use (in RAM) as few as $n$ bits on top of a packed representation of the input/output, and process data as fast as $2.92$ megabases per second.


Contributions (\ref{goalLCP}, \ref{goalMerge}) are part of a preliminary version~\cite{prezza_et_al:LIPIcs:2019:10478} of this paper.
This paper also extends such results 
with the suffix tree interval enumeration procedure and 
with the algorithms of contribution (\ref{goalPLCP}) for building the PLCP bitvector and the BPS representation of the suffix tree topology.

\section{Basic Concepts}\label{sec:notation}

Let $\Sigma =\{c_1, c_2, \ldots, c_\sigma\}$ be a finite ordered alphabet of size $\sigma$ with $\# = c_1< c_2< \ldots < c_\sigma$, where $<$ denotes the standard lexicographic order.
Given a text $T=t_1 t_2 \cdots t_n \in \Sigma^*$ we denote by $|T|$ its length $n$. 
We assume that the input text is terminated by the special symbol (terminator) $\#$, which does not appear elsewhere in $T$. 
We use $\epsilon$ to denote the empty string.
A \emph{factor} (or \emph{substring}) of $T$ is written as $T[i,j] = t_i \cdots t_j$ with $1\leq i \leq j \leq n$.  When declaring an array $A$, we use the same notation $A[1,n]$ to indicate that the array has $n$ entries indexed from $1$ to $n$.
A \emph{right-maximal} substring $W$ of $T$ is a string for which there exist at least two distinct characters $a,b$ such that $Wa$ and $Wb$ occur in $T$.

The \emph{suffix array} SA of a string $T$ (see \cite{PuglisiTurpin2008} for a survey) is an array containing the permutation of the integers $1,2, \ldots, n$ that arranges the starting positions of the suffixes of $T$ into lexicographical order, i.e., for all $1 \leq i < j \leq n$, $SA[i] < SA[j]$.

The \emph{inverse suffix array} $ISA[1, n]$ is the inverse permutation of $SA$, i.e., $ISA[i] = j$ if and only if $SA[j] = i$. 

The Burrows-Wheeler Transform of a string $T$ is a reversible transformation  that permutates its symbols, i.e. $BWT[i]=T[SA[i]-1]$ if $SA[i] > 1$ or $\#$ otherwise.

In some of our results we deal with \emph{string collections}. 
There exist some natural extensions of the suffix array and the Burrows-Wheeler Transform to a collection of strings.

Let $\bigS = \{T_1, \dots, T_m\}$ be a string collection of total length $n$, where each $T_i$ is terminated by a character $\#$ (the terminator) lexicographically smaller than all other alphabet's characters. In particular, a collection is an ordered multiset, and we denote $\bigS[i] = T_i$.

We define lexicographic order among the strings' suffixes in the usual way, except that, \emph{only while sorting}, each terminator $\#$ of the $i$-th string $\bigS[i]$ is considered (implicitly) a different symbol $\#_i$, with $\#_i < \#_j$ if and only if $i<j$.
Equivalently, in case of equal suffixes ties are broken by input's order: if $T_i[k,|T_i|-1]=T_j[k',|T_j|-1]$, then we define $T_i[k,|T_i|] < T_j[k',|T_j|]$ if and only if $i < j$.

The \emph{generalized suffix array} $GSA[1,n]$ (see~\cite{Shi:1996,CGRS_JDA_2016,Louza2017}) of  $\bigS$ is an array of pairs $GSA[i] = \langle j,k \rangle$ such that 
$\bigS[j][k,|\bigS[j]|]$
is the $i$-th lexicographically smallest suffix of strings in $\bigS$, where we break ties by input position (i.e. $j$ in the notation above).
Note that, if the collection is formed by a single string $T$, then the first component in $GSA$'s pairs is always equal to 1, and the second components form the suffix array of $T$.
We denote by $\mathtt{range(W)} = \langle \mathtt{left(W)}, \mathtt{right(W)} \rangle$, also referred to as \emph{suffix array (SA) interval of $W$, or simply $W$-interval}, the maximal pair $\langle L,R \rangle$ such that all suffixes in $GSA[L,R]$ are prefixed by $W$. We use the same notation with the suffix array of a single string $T$. 
Note that the number of suffixes lexicographically smaller than $W$ in the collection is $L-1$. 
We extend this definition also to cases where $W$ is not present in the collection: in this case, the (empty) range is $\langle L, L-1\rangle$ and we still require that $L-1$ is the number of suffixes lexicographically smaller than $W$ in the collection (or in the string). 

The \emph{extended Burrows-Wheeler Transform}
$BWT[1,n]$ \cite{MantaciRRS07,BauerCoxRosoneTCS2013} of $\bigS$ is the character array defined as $BWT[i] = \bigS[j][k-1\ \mathtt{mod}\ |\bigS[j]|]$, where $\langle j,k \rangle =  GSA[i]$. 

To simplify notation, we indicate with ``$BWT$'' both the Burrows-Wheeler Transform of a string and of a string collection. The used transform will be clear from the context.

The \emph{longest common prefix} (LCP) array of a string $s$ \cite{ManberMyers1993} (resp. a collection $\bigS$ of strings, see \cite{CGRS_JDA_2016,Louza2017,EgidiAMB2019}) is an array storing the length of the longest common prefixes between two consecutive suffixes of $s$ (resp. $\bigS$) in lexicographic order (with $LCP[1]=0$). When applied to a string collection, we take the longest common prefix of two equal suffixes of length $\ell$ to be equal to $\ell-1$ (i.e. as if their terminators were different).

Given two collections $\bigS_1, \bigS_2$ of total length $n$, the Document Array of their union is the binary array $DA[1,n]$ such that $DA[i] = 0$ if and only if the $i$-th smallest suffix comes from $\bigS_1$. When merging suffixes of the two collections, ties are broken by collection number (i.e. suffixes of $\bigS_1$ are smaller than suffixes of $\bigS_2$ in case of ties).

The $C$-array of a string (or collection) $S$ is an array $C[1,\sigma]$ such that $C[i]$ contains the number of characters lexicographically smaller than $i$ in $S$, plus one ($S$ will be clear from the context). 
Equivalently, $C[c]$ is the starting position of suffixes starting with $c$ in the suffix array of the string. 
When $S$ (or any of its permutations) is represented with a balanced wavelet tree, then we do not need to store explicitly $C$, and $C[c]$ can be computed in $O(\log\sigma)$ time with no space overhead on top of the wavelet tree (see~\cite{navarro2012wavelet}). 
Function $\mathtt{S.rank_c(i)}$ returns the number of characters equal to $c$ in $S[1,i-1]$. When $S$ is represented by a wavelet tree, \emph{rank} can be computed in $O(\log \sigma)$ time.

Function $\mathtt{getIntervals(L,R,BWT)}$, where $BWT$ is the extended Burrows-Wheeler transform of a string collection $\bigS$ and $\langle L,R\rangle$ is the 
suffix array interval 
of some string
$W$ appearing as a substring of some element of $\bigS$, returns all suffix array intervals of strings $cW$, with $c\neq \#$, that occur in $\bigS$. 
When $BWT$ is represented with a balanced wavelet tree, we can implement this function so that it terminates in $O(\log\sigma)$ time per returned interval~\cite{beller2013computing}. The function can be made to return the output intervals on-the-fly, one by one (in an arbitrary order), without the need to store them all in an auxiliary vector, with just $O(\log n)$ bits of additional overhead in space~\cite{beller2013computing} (this requires to DFS-visit the sub-tree of the wavelet tree induced by $BWT[L,R]$; the visit requires only $\log\sigma$ bits to store the current path in the tree).

An extension of the above function that navigates in parallel two BWTs is immediate. Function $\mathtt{getIntervals(L_1,R_1,L_2, R_2, BWT_1, BWT_2)}$ takes as input two ranges of a string $W$ on the BWTs of two collections, and returns the pairs of ranges on the two BWTs corresponding to all left-extensions $cW$ of $W$ ($c\neq \#$) such that $cW$ appears in at least one of the two collections. To implement this function, it is sufficient to navigate in parallel the two wavelet trees as long as at least one of the two intervals is not empty.

Let $S$ be a string. 
The function $S.\mathtt{rangeDistinct(i,j)}$ returns the set of distinct alphabet characters \emph{different than the terminator} $\#$ in $S[i,j]$. Also this function can be implemented in $O(\log\sigma)$ time per returned element when $S$ is represented with a wavelet tree (again, this requires a DFS-visit of the sub-tree of the wavelet tree induced by $S[i,j]$).

$BWT.\mathtt{bwsearch(\langle L,R \rangle, c)}$ is the function that, given the suffix array interval $\langle L,R \rangle$ of a string $W$ occurring in the collection, returns the suffix array interval of $cW$ by using the BWT of the collection~\cite{ferragina2000opportunistic}. This function requires access to array $C$ and \emph{rank} support on $BWT$, and runs in $O(\log\sigma)$ time when $BWT$ is represented with a balanced wavelet tree.

To conclude, our algorithms will take as input a wavelet tree representing the BWT. As shown in the next lemma by Claude et al., this is not a restriction:

\begin{lemma}[\cite{claude2015wavelet}]\label{thm:BWT->WT}
	Given a word-packed string of length $n$ on alphabet $[1,\sigma]$, we can replace it with its wavelet matrix~\cite{claude2015wavelet} in $O(n\log\sigma)$ time using $n$ bits of additional working space.
\end{lemma}

Wavelet matrices~\cite{claude2015wavelet} are a particular space-efficient representation of wavelet trees taking $n\log\sigma \cdot (1+o(1))$ bits of space and supporting all their operations within the same running times. Since the output of all our algorithms will take at least $n$ bits, it will always be possible to re-use a portion of the output's space (before computing it) to fit the extra $n$ bits required by Lemma \ref{thm:BWT->WT}.


\section{Belazzougui's Enumeration Algorithm}\label{sec:belazzougui}

In~\cite{Belazzougui:2014:LTC:2591796.2591885}, Belazzougui showed that a BWT with \emph{rank} and \emph{range distinct} functionality (see Section \ref{sec:notation}) is sufficient to enumerate in small space a rich representation of the internal nodes of the suffix tree of a text $T$. For the purposes of this article, we assume that the BWT is represented using a wavelet tree (whereas Belazzougui's original result is more general), and thus that all queries take $O(\log \sigma)$ time.


\begin{theorem}[Belazzougui \cite{Belazzougui:2014:LTC:2591796.2591885}]\label{th:Belazzougui}
	Given the Burrows-Wheeler Transform of a text $T\in[1,\sigma]^n$ represented with a wavelet tree, we can enumerate the following information for each distinct right-maximal substring $W$ of $T$: (i) $|W|$, and (ii) $range(Wc_i)$ for all $c_1 < \dots < c_k$ such that $Wc_i$ occurs in $T$. The process runs in $O(n\log\sigma)$ time and uses $O(\sigma^2\log^2n)$ bits of working space on top of the BWT. 
\end{theorem}

To keep the article self-contained, in this section we describe the algorithm at the core of the above result. 
Remember that explicit suffix tree nodes correspond to right-maximal substrings. 
The first idea is to represent any substring $W$ (not necessarily right-maximal) as follows. Let $\mathtt{chars_W[1,k_W]}$ be the alphabetically-sorted character array such that $W\cdot \mathtt{chars_W[i]}$ is a substring of $T$ for all $i=1,\dots, k_W$, where $k_W$ is the number of right-extensions of $W$.
We require $\mathtt{chars_W}$ to be also complete: if $Wc$ is a substring of $T$, then $c\in \mathtt{chars_W}$.
Let moreover $\mathtt{first_W[1,k_W+1]}$ be the array such that $\mathtt{first_W[i]}$ is the starting position of (the range of) $W\cdot \mathtt{chars_W[i]}$ in the suffix array of $T$ for $i=1,\dots, k_W$, and $\mathtt{first_W[k_W+1]}$ is the end position of $W$ in the suffix array of $T$. The representation for $W$ is (differently from~\cite{Belazzougui:2014:LTC:2591796.2591885}, we omit $\mathtt{chars_W}$ from the representation and we add $|W|$; these modifications will turn useful later):
$$\mathtt{repr(W) = \langle \mathtt{first_W},\ |W| \rangle}$$

Note that, if $W$ is not right-maximal nor a text suffix, then $W$ is followed by $k_W=1$ distinct characters in $T$ and the above representation is still well-defined.
When $W$ is right-maximal, we will also say that $\mathtt{repr(W)}$ is the representation of a suffix tree explicit node (i.e. the node reached by following the path labeled $W$ from the root).

\paragraph{Weiner Link Tree Visit}\label{app:belazzougui}

The enumeration algorithm works by visiting the Weiner Link tree of $T$ starting from the root's representation, that is, $\mathtt{repr(\epsilon) = \langle \mathtt{first_\epsilon},\ 0 \rangle}$, where $\mathtt{first_\epsilon} = \langle C[c_1], \dots, C[c_\sigma], n \rangle$ (see Section \ref{sec:notation} for a definition of the $C$-array) and $c_1, \dots, c_\sigma$ 
are  the sorted alphabet's characters.
Since the suffix tree and the Weiner link tree share the same set of nodes, this is sufficient to enumerate all suffix tree nodes. 
The visit uses a stack storing representations of suffix tree nodes, initialized with  $\mathtt{repr(\epsilon)}$. At each iteration, we pop the head $\mathtt{repr(W)}$ from the stack and we push $\mathtt{repr(cW)}$ such that $cW$ is right-maximal in $T$. 
To keep the stack's size under control,  once computed $\mathtt{repr(cW)}$ for the right-maximal left-extensions $cW$ of $W$ we push them on the stack in decreasing order of range length $\mathtt{range(cW)}$ (i.e. the node with the smallest range is pushed last). This guarantees that the stack will always contain at most $O(\sigma\log n)$ elements~\cite{Belazzougui:2014:LTC:2591796.2591885}. Since each element takes $O(\sigma\log n)$ bits to be represented, the stack's size never exceeds $O(\sigma^2\log^2 n)$ bits. 

\paragraph{Computing Weiner Links}

We now show how to efficiently compute the node representation $\mathtt{repr(cW)}$ from $\mathtt{repr(W)}$ for the characters $c$ such that $cW$ is right-maximal in $T$. In~\cite{Belazzougui:2014:LTC:2591796.2591885,belazzougui2016linear}  this operation is supported efficiently by first enumerating all \emph{distinct} characters in each range $BWT[\mathtt{first_W[i], first_W[i+1]}]$ for $i=1, \dots, k_W$, using function $\mathtt{BWT.rangeDistinct(first_W[i], first_W[i+1])}$ (see Section \ref{sec:notation}). 
Equivalently, for each $a\in \mathtt{chars_W}$ we want to list all distinct left-extensions $cWa$ of $Wa$.
Note that, in this way, we may also visit implicit suffix tree nodes (i.e. some of these left-extensions could be not right-maximal).
Stated otherwise, we are traversing all explicit \emph{and} implicit Weiner links. Since the number of such links is linear~\cite{Belazzougui:2014:LTC:2591796.2591885,belazzougui2014alphabet} (even including implicit Weiner links\footnote{To see this, first note that the number of right-extensions $Wa$ of $W$ that have only one left-extension $cWa$ is at most equal to the number of right-extensions of $W$; globally, this is at most the number of suffix tree's nodes (linear). Any other right-extension $Wa$ that has at least two distinct left-extensions $cWa$ and $bWa$ is, by definition, left maximal and corresponds therefore to a node in the suffix tree of the reverse of $T$. It follows that all left-extensions of $Wa$ can be charged to an edge of the suffix tree of the reverse of $T$ (again, the number of such edges is linear).}), globally the number of distinct characters returned by $\mathtt{rangeDistinct}$ operations is $O(n)$. An implementation of $\mathtt{rangeDistinct}$ on wavelet trees is discussed in \cite{beller2013computing} with the procedure \texttt{getIntervals} (this procedure actually returns more information: the suffix array range of each $cWa$). This implementation runs in $O(\log\sigma)$ time per returned character. Globally, we therefore spend $O(n\log\sigma)$ time using a wavelet tree. We now need to compute $\mathtt{repr(cW)}$ for all left-extensions of $W$ and keep only the right-maximal ones. Let $x=\mathtt{repr(W)}$ and $\mathtt{BWT.Weiner(x)}$ be the function that returns the representations of such strings (used in Line \ref{range distinct2} of Algorithm \ref{alg:fill nodes}).
This function can be implemented by observing that
$$
\begin{array}{lcl}
     \mathtt{range(cWa)} & = \mathtt{\langle}&  \mathtt{C[c] + BWT.rank_c(left(Wa))}, \\
     && \mathtt{ C[c] + BWT.rank_c(right(Wa)+1)-1 \ \rangle} \\
\end{array}
$$
where $a=\mathtt{chars_W[i]}$ for $1\leq i < |\mathtt{first_W}|$, and noting that $\mathtt{left(Wa)}$ and $\mathtt{right(Wa)}$ are available in $\mathtt{repr(W)}$. Note also that we do not actually need to know the value of characters $\mathtt{chars_W[i]}$ to compute the ranges of each $cW\cdot \mathtt{chars_W[i]}$; this is the reason why we can omit $\mathtt{chars_W}$ from $\mathtt{repr(W)}$.
Using a wavelet tree, the above operation takes $O(\log\sigma)$ time. By the above observations, the number of strings $cWa$ such that $W$ is right-maximal is bounded by $O(n)$. Overall, computing $\mathtt{repr(cW)} = \langle \mathtt{first_{cW}}, |W|+1 \rangle$ for all left-extensions $cW$ of all right-maximal strings $W$ takes therefore $O(n\log\sigma)$ time. Within the same running time, we can check which of those extensions is right maximal (i.e. those such that $|\mathtt{first_{cW}}|\geq 2$), sort them in-place by interval length (we always sort at most $\sigma$ node representations, therefore also sorting takes globally $O(n\log\sigma)$ time), and push them on the stack.

\section{Beller et al.'s Algorithm}\label{sec:beller}

The second ingredient used in our solutions is the following result, due to Beller et al (we slightly re-formulate their result to fit our purposes, read below for a description of the differences):

\begin{theorem}[Beller et al.\cite{beller2013computing}]\label{th:Beller}
	Given the Burrows-Wheeler Transform of a text $T$ represented with a wavelet tree, we can enumerate all pairs $(i,LCP[i])$ in $O(n\log\sigma)$ time using $5n$ bits of working space on top of the BWT.
\end{theorem}

Theorem \ref{th:Beller} represents the state of the art for computing the LCP array from the BWT. 
Also Beller et al.'s algorithm works by enumerating a (linear) subset of the BWT intervals. LCP values are induced from a particular visit of those intervals. Belazzougui's and Beller et al.'s algorithms have, however, two key differences which make the former more space-efficient on small alphabets, while the latter more space-efficient on large alphabets: 	(i) Beller et al. use a queue (FIFO) instead of a stack (LIFO), and (ii) they represent $W$-intervals with just the pair of coordinates $\mathtt{range(W)}$ and the value $|W|$.
In short, while Beller et al.'s queue might grow up to size $\Theta(n)$, the use of intervals (instead of the more complex representation used by Belazzougui) makes it possible to represent it using $O(1)$ bitvectors of length $n$. On the other hand, the size of Belazzougui's stack can be upper-bounded by $O(\sigma\log n)$, but its elements take more space to be represented. 

We now describe in detail Beller et al.'s result.
We keep a bitvector $U[1,n]$ such that $U[i]=0$ if and only if the pair $(i,LCP[i])$ has not been output yet. In their original algorithm, Beller et al. use the LCP array itself to mark undefined LCP entries. In our case, we don't want to store the whole LCP array (for reasons that will be clear in the next sections) and thus we only record which LCP values have been output. Bitvector $U$ accounts for the additional $n$ bits used by Theorem \ref{th:Beller} with respect to the original result described in~\cite{beller2013computing}.
At the beginning, $U[i]=0$ for all $i=1, \dots, n$.
Beller et al.'s algorithm starts by inserting in the queue the triple $\langle1,n,0 \rangle$, where the first two components are the BWT interval of $\epsilon$ (the empty string) and the third component is its length. From this point, the algorithm keeps performing the following operations until the queue is empty. We remove the first (i.e. the oldest) element $\langle L,R,\ell\rangle$ from the queue, which (by induction) is the interval and length of some string $W$: $\mathtt{range(W)}= \langle L,R \rangle$ and $|W|=\ell$. 
Using  operation $\mathtt{getIntervals(L,R,BWT)}$~\cite{beller2013computing} (see Section \ref{sec:notation})
we left-extend the BWT interval $\langle L,R\rangle$ with the characters $c_1, \dots, c_k$ in $\mathtt{BWT.rangeDistinct(L,R)}$, obtaining the triples $\langle L_1, R_1, \ell+1 \rangle, \dots, \langle L_k, R_k, \ell+1 \rangle$ corresponding to the strings $c_1W, \dots, c_kW$. For each such triple $\langle L_i, R_i, \ell+1\rangle$, if $R_i\neq n$ and $U[R_i+1] = 0$ then we set $U[R_i+1] \leftarrow 1$, we output the LCP pair $(R_i+1, \ell)$ and push $\langle L_i, R_i, \ell+1\rangle$ on the queue. Importantly, note that we can push the intervals returned by $\mathtt{getIntervals(L,R,BWT)}$ in the queue in any order; as discussed in Section \ref{sec:notation}, this step can be implemented with just $O(\log n)$ bits of space overhead with a DFS-visit of the wavelet tree's sub-tree induced by $BWT[L,R]$ (i.e. the intervals are not stored temporarily anywhere: they are pushed as soon as they are generated).

\paragraph{Queue implementation} To limit space usage, Beller et al. use the following queue representations. First note that, at each time point, the queue's triples are partitioned into a (possibly empty) sequence with associated length (i.e. the third element in the triples) $\ell+1$, followed by a sequence with associated length $\ell$, for some $\ell$. 
To simplify the description, let us assume that these two sequences are kept as two distinct queues, indicated in the following as $Q_\ell$ and $Q_{\ell+1}$. At any stage of the algorithm, we pop from $Q_\ell$ and push into $Q_{\ell+1}$. It follows that there is no need to store strings' lengths in the triples themselves (i.e. the queue's elements become just ranges), since the length of each element in $Q_\ell$ is $\ell$.
When $Q_\ell$ is empty, we create a new empty queue $Q_{\ell+2}$, pop from $Q_{\ell+1}$, and push into $Q_{\ell+2}$ (and so on).
Beller et al. represent $Q_\ell$ as follows. While pushing elements in $Q_\ell$, as long as its size does not exceed $n/\log n$ we represent it as a vector of pairs (of total size at most $O(n)$ bits). This representation supports push/pop operations in (amortized) constant time and takes at most $O(\log n \cdot n/\log n) = O(n)$ bits of space. 
As soon as $Q_\ell$'s size exceeds $n/\log n$, we switch to a representation that uses two packed bitvectors of length $n$ storing, respectively, the left- and right-most boundaries of the ranges in the queue. 
Note that this representation can be safely used since the pairs in $Q_\ell$ are suffix array ranges of strings of some fixed length $\ell$, therefore there cannot be overlapping intervals.
Pushing an interval into such a queue takes constant time (it just requires setting two bits). Popping all the $t = |Q_\ell|$ intervals, on the other hand, can easily be implemented in $O(t+ n/\log n)$ time by scanning the bitvectors and exploiting word-parallelism (see \cite{beller2013computing} for all details). Since Beller et al.'s procedure visits $O(n)$ SA intervals, $Q_\ell$ will exceed size $n/\log n$ for at most $O(\log n)$ values of $\ell$. It follows that also with this queue representation pop operations take amortized constant time. 

\paragraph{Time complexity} It is easy to see that the algorithm inserts in total a linear number of intervals in the queue since an interval $\langle L_i, R_i, \ell+1 \rangle$ is inserted only if $U[R_i+1] = 0$, and successively $U[R_i+1]$ is set to $1$. Clearly, this can happen at most $n$ times. 
In~\cite{beller2013computing} the authors moreover show that, even when counting the left-extensions of those intervals (computed after popping each interval from the queue), the total number of generated intervals stays linear.
Overall, the algorithm runs therefore in $O(n\log\sigma)$ time (as discussed in Section \ref{sec:notation}, $\mathtt{getIntervals}$ runs in $O(\log\sigma)$ time per returned element).

\section{Enumerating LCP values}\label{sec:LCP}

In this section we prove our first main result: how to enumerate LCP pairs $(i,LCP[i])$ using succinct working space on top of a wavelet tree representing the BWT. 
Later we will use this procedure to build the LCP and PLCP arrays in small space on top of a plain representation of the BWT.
We give our lemma in the general form of string collections, which will require adapting the algorithms seen in the previous sections to this more general setting. Our first observation is that Theorem \ref{th:Belazzougui}, extended to string collections as described below, can be directly used to enumerate LCP pairs $(i,LCP[i])$ using just $O(\sigma^2\log^2n)$ bits of working space on top of the input and output. We combine  this procedure with an extended version of Beller et al.'s algorithm working on string collections in order to get small working space for all alphabets. Algorithms \ref{alg:fill nodes} and \ref{alg:fill leaves} report our complete procedure; read below for an exhaustive description. We obtain our first main result: 

\begin{lemma}\label{thm:LCP collection}
	Given a wavelet tree for the Burrows-Wheeler Transform of a collection $\bigS = \{T_1, \dots, T_m\}$ of total length $n$ on alphabet $[1,\sigma]$, we can enumerate all pairs $(i, LCP[i])$ in $O(n\log\sigma)$ time using $o(n\log\sigma)$ bits of working space on top of the BWT.
\end{lemma}
\begin{proof}
    If $\sigma < \sqrt{n}/\log^2n$ then $\sigma^2\log^2n = o(n)$ and our extension of Theorem \ref{th:Belazzougui} gives us $o(n\log\sigma)$ additional working space. If $\sigma \geq \sqrt{n}/\log^2n$ then $\log\sigma = \Theta(\log n)$ and we can use our extension to string collections of Theorem \ref{th:Beller}, which yields extra working space $O(n) = o(n\log n) = o(n\log\sigma)$. Note that, while we used the threshold $\sigma < \sqrt{n}/\log^2n$, any threshold of the form $\sigma < \sqrt{n}/\log^{1+\epsilon}n$, with $\epsilon>0$ would work. The only constraint is that $\epsilon>0$, since otherwise for $\epsilon=0$ the working space would become $O(n\log\sigma)$ for constant $\sigma$ (not good since we aim at $o(n\log\sigma)$).  
\end{proof}

We now describe all the details of our extensions of Theorems \ref{th:Belazzougui} and \ref{th:Beller} used in the proof of Lemma \ref{thm:LCP collection}.
Procedure \texttt{BGOS(BWT)} in Line \ref{beller et al.} of Algorithm \ref{alg:fill nodes} is a call to Beller et al.'s algorithm, modified as follows. First, we enumerate the LCP pairs $(C[c], 0)$ for all $c\in\Sigma$. Then, we push in the queue $\langle \mathtt{range(c), 1} \rangle$ for all $c\in\Sigma$ and start the main algorithm. Note moreover that (see Section \ref{sec:notation}) from now on we never left-extend ranges with $\#$.

Recall that each string of a text collection $\bigS$ is ended by a terminator $\#$ (common to all strings). 
Consider now the LCP and GSA 
arrays of $\bigS$. We divide LCP values in two types. 
Let $GSA[i] = \langle j,k \rangle$, with $i>1$, indicate that the $i$-th suffix in the lexicographic ordering of all suffixes of strings in $\bigS$ is $\bigS[j][k,|\bigS[j]|]$. A LCP value $\mathtt{LCP[i]}$ is of \emph{node type} when the $i$-th and $(i-1)$-th suffixes are distinct: $\bigS[j][k,|\bigS[j]|] \neq \bigS[j'][k',|\bigS[j']|]$, where $GSA[i] = \langle j,k \rangle$ and $GSA[i-1] = \langle j',k' \rangle$. Those two suffixes differ before the terminator is reached in both suffixes (it might be reached in one of the two suffixes, however); we use the name \emph{node-type} because $i-1$ and $i$ are the last and first suffix array positions of the ranges of two adjacent children of some suffix tree node, respectively (i.e. the node corresponding to string $\bigS[j][k,k+LCP[i]-1]$).
Note that it might be that one of the two suffixes, $\bigS[j][k,|\bigS[j]|]$ or $\bigS[j'][k',|\bigS[j']|]$, is the string ``$\#$''.
Similarly, a \emph{leaf-type} LCP value $\mathtt{LCP[i]}$ is such that the $i$-th and $(i-1)$-th suffixes are equal: $\bigS[j][k,|\bigS[j]|] = \bigS[j'][k',|\bigS[j']|]$. We use the name \emph{leaf-type} because, in this case, it must be the case that $i \in [L+1,R]$, where $\langle L,R \rangle$ is the suffix array range of some suffix tree leaf (it might be that $R>L$ since there might be repeated suffixes in the collection). Note that, in this case, $\bigS[j][k,|\bigS[j]|] = \bigS[j'][k',|\bigS[j']|]$ could coincide with $\#$. Entry $LCP[0]$ escapes the above classification, so we output it separately.

Our idea is to compute first node-type and then leaf-type LCP values. We argue that Beller et al.'s algorithm already computes the former kind of LCP values. When this algorithm uses too much space (i.e. on small alphabets), we show that Belazzougui's enumeration strategy can be adapted to reach the same goal: by the very definition of node-type LCP values, they lie between children of some suffix tree node $x$, and their value corresponds to the string depth of $x$. This strategy is described in Algorithm \ref{alg:fill nodes}. Function $\mathtt{BWT.Weiner(x)}$ in Line \ref{range distinct2} takes as input the representation of a suffix tree node $x$ and returns all explicit nodes reached by following Weiner links from $x$ (an implementation of this function is described in Section~\ref{sec:belazzougui}).
Leaf-type LCP values, on the other hand, can easily be computed by enumerating intervals corresponding to suffix tree leaves. To reach this goal, it is sufficient to enumerate ranges of suffix tree leaves starting from $\mathtt{range(\#)}$ and recursively left-extending with backward search with characters different from $\#$ whenever possible. For each range $\langle L,R \rangle$ obtained in this way, we set each entry $LCP[L+1,R]$ to the string depth (terminator excluded) of the corresponding leaf. 
This strategy is described in Algorithm \ref{alg:fill leaves}.
In order to limit space usage, we use again a stack or a queue to store leaves and their string depth (note that each leaf takes $O(\log n)$ bits to be represented): we use a queue when $\sigma > n/\log^3n$, and a stack otherwise.
The queue is the same used by Beller et al.\cite{beller2013computing} and described in Section \ref{sec:beller}.
This guarantees that the bit-size of the queue/stack never exceeds $o(n\log\sigma)$ bits: since leaves take just $O(\log n)$ bits to be represented and the stack's size never contains more than $O(\sigma\cdot \log n)$ leaves, the stack's bit-size never exceeds $O(n/\log n) = o(n)$ when $\sigma \leq n/\log^3n$. Similarly, Beller et al's queue always takes at most $O(n)$ bits of space, which is $o(n\log\sigma)$ for $\sigma > n/\log^3n$.
Note that in Lines \ref{getIntervals1}-\ref{push2} we can afford storing temporarily the $k$ resulting intervals since, in this case, the alphabet's size is small enough. 

To sum up, our full procedure works as follows: (1), we output node-type LCP values using procedure \texttt{Node-Type}$(\mathtt{BWT})$ described in Algorithm \ref{alg:fill nodes}, and (2) we output leaf-type LCP values using procedure \texttt{Leaf-Type}$(\mathtt{BWT})$ described in Algorithm \ref{alg:fill leaves}.

\begin{algorithm}
  \begin{algorithmic}[1]
  
	\If{$\sigma > \sqrt n/\log^2n$}
	  \State $\mathtt{BGOS(BWT)}$ \Comment Run Beller et al.'s algorithm\label{beller et al.}
	\Else 
	  \State $\mathtt P \leftarrow \texttt{new\_stack()}$\Comment Initialize new stack\label{new stack2}
	  \State $\mathtt P\mathtt{.push( repr(\epsilon))}$\Comment Push representation of $\epsilon$ \label{push3}
	  \While{$\mathtt{\mathbf{not}\ P.empty()}$}\label{while2}
	  
	     \State $\langle \mathtt{first_W},\ \ell \rangle \leftarrow \mathtt{P.pop()}$\Comment Pop highest-priority element\label{pop2}	
         \State $t \leftarrow |\mathtt{first_W}|-1$\Comment Number of children of ST node\label{nchild}
	  
	     \For{$i = 2, \dots, t$}
	        \State \textbf{output} $(\mathtt{first_W}[i], \ell)$\Comment Output LCP value\label{LCP in Node}
         \EndFor	  
         
         \State $x_1,\dots, x_k \leftarrow \mathtt{BWT.Weiner(\langle first_W,\ \ell \rangle)}$\Comment Follow Weiner Links\label{range distinct2}
		 \State $x'_1, \dots, x'_k \leftarrow \mathtt{sort}(x_1, \dots, x_k)$\Comment Sort by interval length\label{sort2}
         
         \For{$i=k\dots 1$}
				
            \State $\mathtt{P.push}(x'_i)$\Comment Push representations\label{push4}

         \EndFor
         
	  \EndWhile
	  
	\EndIf

 \caption{\texttt{Node-Type(BWT)}}\label{alg:fill nodes}
\end{algorithmic}
\end{algorithm}

\begin{algorithm}
  \begin{algorithmic}[1]
  
       \For{$i=left(\#),\dots, right(\#)$}
        \State \textbf{output} $(i,0)$
      \EndFor
      
     \If{$\sigma > n/\log^3n$}

       \State $\mathtt P \leftarrow \mathtt{new\_queue()}$\Comment{Initialize new queue}\label{new queue1}

     \Else
     
       \State $\mathtt P \leftarrow \mathtt{new\_stack()}$\Comment{Initialize new stack}\label{new stack1}

     \EndIf		

	\State $\mathtt P\mathtt{.push( BWT.range(\#),0)}$\Comment{Push range of terminator and LCP value 0}\label{push1}

    \While{$\mathtt{\mathbf{not}\ P.empty()}$}\label{while1}

     \State $\langle \langle L,R \rangle, \ell \rangle \leftarrow \mathtt{P.pop()}$\Comment{Pop highest-priority element}\label{pop1}	
  
     \For{$i=L+1\dots R$}
        \State \textbf{output} $(i, \ell)$\Comment{Output LCP inside range of ST leaf}\label{LCP in Leaves}
      \EndFor

	\If{$\sigma > n/\log^3n$}
  
       \State $\mathtt{P.push(getIntervals(L, R, BWT), \ell+1)}$\Comment{Pairs $\langle$interval,$\ell+1\rangle$}\label{push7}
				
    \Else
      \State $\langle L_i, R_i\rangle_{i=1, \dots, k} \leftarrow \mathtt{getIntervals(L, R,BWT)}$\label{getIntervals1}
      \State $\langle L'_i, R'_i\rangle_{i=1, \dots, k} \leftarrow \mathtt{sort}(\langle L_i, R_i\rangle_{i=1, \dots, k})$\Comment{Sort by interval length}\label{sort1}

      \For{$i=k\dots 1$}
					
		\State $\mathtt{P.push}(\langle L'_i, R'_i\rangle, \ell+1)$\Comment{Push in order of decreasing length}\label{push2}

      \EndFor

    \EndIf

    \EndWhile

 \caption{\texttt{Leaf-Type(BWT)}}\label{alg:fill leaves}
\end{algorithmic}
\end{algorithm}

The correctness, completeness, and complexity of our procedure are proved in the following Lemma: 

\begin{lemma}\label{lemma:proof of thm1}
	Algorithms \ref{alg:fill nodes} and \ref{alg:fill leaves} correctly output all LCP pairs $(i,LCP[i])$ of the collection in $O(n\log\sigma)$ time using $o(n\log\sigma)$ bits of working space on top of the input BWT.
\end{lemma}
\begin{proof}
	
	\emph{Correctness - Algorithm \ref{alg:fill nodes}}. We start by proving that Beller et al.'s procedure in Line \ref{beller et al.} of Algorithm \ref{alg:fill nodes} (procedure \texttt{BGOS(BWT)}) outputs all the node-type LCP entries correctly. The proof  proceeds by induction on the LCP value $\ell$ and follows the original proof of~\cite{beller2013computing}.
	At the beginning, we insert in the queue all $c$-intervals, for $c\in\Sigma$. For each such interval $\langle L,R \rangle$ we output $LCP[R+1]=\ell = 0$. It is easy to see that after this step all and only the node-type LCP values equal to 0 have been correctly computed. 
	Assume, by induction, that all node-type LCP values less than or equal to $\ell$ have been correctly output, and that we are about to extract from the queue the first triple $\langle L,R,\ell+1 \rangle$ having length $\ell+1$. For each extracted triple with length $\ell+1$ associated to a string $W$, consider the triple $\langle L',R',\ell+2 \rangle$ associated to one of its left-extensions $cW$. If $LCP[R'+1]$ has been computed, i.e. if $U[R'+1]=1$, then we have nothing to do. However, if $U[R'+1]=0$, then it must be the case that (i) the corresponding LCP value satisfies $LCP[R'+1] \geq \ell+1$, since by induction we have already computed all node-type LCP values smaller than or equal to $\ell$, and (ii) $LCP[R'+1]$ is of node-type, since otherwise the BWT interval of $cW$ would also include position $R'+1$. On the other hand, it cannot be the case that $LCP[R'+1] > \ell+1$ since otherwise the $cW$-interval would include position $R'+1$. We therefore conclude that $LCP[R'+1] = \ell+1$ must hold.

	\emph{Completeness - Algorithm \ref{alg:fill nodes}}. The above argument settles correctness; to prove completeness, assume that, at some point, $U[i] = 0$ and the value of $LCP[i]$ to be computed and output is $\ell+1$. We want to show that we will pull a triple $\langle L,R,\ell+1 \rangle$ from the queue corresponding to a string $W$ (note that $\ell+1=|W|$ and, moreover, $W$ could end with $\#$) such that one of the left-extensions $aW$ of $W$ satisfies $\mathtt{range(aW)} = \langle L',i-1 \rangle$, for some $L'$. This will show that, at some point, we will output the LCP pair $(i,  \ell+1)$. 
	We proceed by induction on $|W|$.
	Note that we separately output all LCP values equal to 0.
	The base case $|W|=1$ is easy: by the way we initialized the queue, $\langle \mathtt{range(c)}, 1\rangle$, for all $c\in\Sigma$, are the first triples we pop. Since we left-extend these ranges with all alphabet's characters except $\#$, it is easy to see that all LCP values equal to 1 have been output. From now on we can therefore assume that we are working on LCP values equal to $\ell+1>1$, i.e. $W=b\cdot V$, for $b\in\Sigma-\{\#\}$ and $V\in \Sigma^+$.
	Let $abV$ be 
	the length-$(\ell+2)$ left-extension of $W=bV$ 
	such that $\mathtt{right(abV)+1} = i$. Since, by our initial hypothesis, $\mathtt{LCP[i]} = \ell+1$, the collection contains also a suffix $aU$ lexicographically larger than $abV$ and such that $\mathtt{LCP(aU,abV)} = \ell+1$. But then, it must be the case that $\mathtt{LCP(right(bV)+1)} = \ell$ (it cannot be smaller by the existence of $U$ and it cannot be larger since $|bV|=\ell+1$). By inductive hypothesis, this value was set after popping a triple $\langle L'', R'', \ell\rangle$ corresponding to string $V$, left-extending $V$ with $b$, and pushing $\langle \mathtt{range(bV)}, \ell+1 \rangle$ in the queue. 
	This ends the completeness proof since we showed that $\langle \mathtt{range(bV)}, \ell+1 \rangle$ is in the queue, so at some point we will pop it, extend it with $a$, and output $(right(abV)+1,\ell+1) = (i,\ell+1)$.
	If the queue uses too much space, then Algorithm \ref{alg:fill nodes} switches to a stack and Lines \ref{new stack2}-\ref{push4} are executed instead of Line \ref{beller et al.}. Note that this pseudocode fragment corresponds to Belazzougui's enumeration algorithm, except that now we also set LCP values in Line \ref{LCP in Node}. By the enumeration procedure's correctness, we have that, in Line \ref{LCP in Node}, $\langle \mathtt{first_W[1]}, \mathtt{first_W[t+1]} \rangle$ is the SA-range of a right-maximal string $W$ with $\ell = |W|$, and $\mathtt{first_W[i]}$ is the first position of the SA-range of $Wc_i$, with $i=1,\dots,t$, where $c_1, \dots, c_2$ are all the (sorted) right-extensions of $W$. Then, clearly each LCP value in Line \ref{LCP in Node} is of node-type and has value $\ell$, since it is the LCP between two strings prefixed by $W\cdot \mathtt{chars_W[i-1]}$ and $W\cdot \mathtt{chars_W[i]}$. Similarly,  completeness of the procedure follows from the completeness of the enumeration algorithm. Let $LCP[i]$ be of node-type. Consider the prefix $Wb$ of length $LCP[i]+1$ of the $i$-th suffix in the lexicographic ordering of all strings' suffixes. Since $LCP[i] = |W|$, the $(i-1)$-th suffix is of the form $Wa$, with $b\neq a$, and $W$ is right-maximal. But then, at some point our enumeration algorithm will visit the representation of $W$, with $|W|=\ell$. Since $i$ is the first position of the range of $Wb$, we have that $i= \mathtt{first_W[j]}$ for some $j \geq 2$, and Line \ref{LCP in Node} correctly outputs the LCP pair $(first_W[j], |W|) = (i,|W|)$.
	
	\emph{Correctness and completeness - Algorithm \ref{alg:fill leaves}}. Proving correctness and completeness of this procedure is much easier. It is sufficient to note that the \texttt{while} loop iterates over all ranges $\langle L,R \rangle$ of strings ending with $\#$ and not containing $\#$ anywhere else (note that we start from the range of $\#$ and we proceed by recursively left-extending this range with symbols different than $\#$). Then, for each such range we conclude that $LCP[L+1,R]$ is equal to $\ell$, i.e. the string depth of the corresponding string (excluding the final character $\#$). By their definition, all leaf-type LCP values are correctly computed in this way.
	
	\emph{Complexity - Algorithm \ref{alg:fill nodes}}.  If $\sigma > \sqrt n/\log^2 n$, then we run Beller et al's algorithm, which terminates in $O(n\log\sigma)$ time and uses $O(n) = o(n\log\sigma)$ bits of additional working space. Otherwise, we perform a linear number of operations on the stack since, as observed in Section \ref{sec:belazzougui}, the number of Weiner links is linear. By the same analysis of Section \ref{sec:belazzougui}, the operation in Line \ref{range distinct2} takes $O(k\log\sigma)$ amortized time on wavelet trees, and sorting in Line \ref{sort2} (using any comparison-sorting algorithm sorting $m$ integers in $O(m\log m)$ time) takes $O(k\log\sigma)$ time. Note that in this sorting step we can afford storing in temporary space nodes $x_1, \dots, x_k$ since this takes additional space $O(k\sigma\log n) = O(\sigma^2\log n) = O(n/\log^3n) = o(n)$ bits. 
	All these operations sum up to $O(n\log\sigma)$ time. Since the stack always takes at most $O(\sigma^2\log^2n)$ bits and $\sigma \leq \sqrt n/\log^2 n$, the stack's size never exceeds $O(n/\log^2n) = o(n)$ bits.
	
	\emph{Complexity - Algorithm \ref{alg:fill leaves}}. Note that, in the \texttt{while} loop, we start from the interval of $\#$ and recursively left-extend with characters different than $\#$ until this is possible. It follows that we visit the intervals of all strings of the form $W\#$ such that $\#$ does not appear inside $W$. Since these intervals form a cover of $[1,n]$, their number (and therefore the number of iterations in the \texttt{while} loop) is also bounded by $n$. This is also the maximum number of operations performed on the queue/stack. Using Beller et al.'s implementation for the queue and a simple vector for the stack, each operation takes constant amortized time. Operating on the stack/queue takes therefore overall $O(n)$ time.
	For each interval $\langle L,R \rangle$ popped from the queue/stack, in Line \ref{LCP in Leaves} we output $R-L-2$ LCP values. As observed above, these intervals form a cover of $[1,n]$ and therefore Line \ref{LCP in Leaves} is executed no more than $n$ times. Line \ref{getIntervals1} takes time $O(k\log\sigma)$. Finally, in Line \ref{sort1} we sort at most $\sigma$ intervals. Using any fast comparison-based sorting algorithm, this costs overall at most $O(n\log\sigma)$ time. 
	
	As far as the space usage of Algorithm \ref{alg:fill leaves} is concerned, note that we always push just pairs interval/length ($O(\log n)$ bits) in the queue/stack. If $\sigma > n/\log^3n$, we use Beller et al.'s queue, taking at most $O(n) = o(n\log\sigma)$ bits of space. Otherwise, the stack's size never exceeds $O(\sigma\cdot \log n)$ elements, with each element taking $O(\log n)$ bits. This amounts to $O(\sigma\cdot \log^2 n) = O(n/\log n) = o(n)$ bits of space usage. Moreover, in Lines \ref{getIntervals1}-\ref{sort1} it holds $\sigma\leq n/\log^3n$ so we can afford storing temporarily all intervals returned by $\mathtt{getIntervals}$ in $O(k\log n) = O(\sigma\log n) = O(n/\log^2n) = o(n)$ bits.
\end{proof}

Combining Lemma \ref{thm:LCP collection} and Lemma \ref{thm:BWT->WT}, we obtain: 

\begin{theorem}\label{thm:LCP collection succinct}
	Given the word-packed Burrows-Wheeler Transform of a collection $\bigS = \{T_1, \dots, T_m\}$ of total length $n$ on alphabet $[1,\sigma]$, we can build the LCP array of the collection in $O(n\log\sigma)$ time using $o(n\log\sigma)$ bits of working space on top of the BWT.
\end{theorem}

\section{Enumerating Suffix Tree Intervals}\label{sec:intervals}

In this section we show that the procedures described in Section \ref{sec:LCP} can be used to enumerate all suffix tree intervals---that is, the suffix array intervals of all right-maximal text substrings---taking as input the BWT of a text. Note that in this section we consider just simple texts rather than string collections as later we will use this procedure to build the compressed suffix tree of a text.

When $\sigma \leq \sqrt n/\log^2n$, we can directly use Belazzougui's procedure (Theorem \ref{th:Belazzougui}), which already solves the problem. For larger alphabets, we modify Beller et al's procedure (Theorem \ref{th:Beller}) to also generate suffix tree's intervals as follows. 

When $\sigma > \sqrt n/\log^2n$, we modify Beller et al.'s procedure to enumerate suffix tree intervals using $O(n) = o(n\log\sigma)$ bits of working space, as follows. 
We recall that (see Section \ref{sec:beller}), Beller et al's procedure can be conveniently described using two separate queues: $Q_{\ell}$ and $Q_{\ell+1}$. At each step, we pop from $Q_\ell$ an element $\langle\langle L,R \rangle, |W| \rangle$ with $\langle L,R \rangle = \mathtt{range(W)}$ and $|W|=\ell$ for some string $W$, left-extend the range with all $a\in \mathtt{BWT.rangeDistinct(L,R)}$, obtaining the ranges $\mathtt{range(aW)} = \langle L_a, R_a\rangle$ and, only if $U[R_a+1]=0$, set $U[R_a+1]\leftarrow 1$, output the LCP pair $(R_a+1, |W|)$, and push $\langle \langle L_a, R_a\rangle, |W|+1 \rangle$ into $Q_{\ell+1}$. Note that, since $LCP[R_a+1] = |W|$ we have that the $R_a$-th and $(R_a+1)$-th  smallest suffixes start, respectively, with $aXc$ and $aXd$ for some $c<d\in\Sigma$, where $W=Xc$. This implies that $aX$ is right-maximal. It is also clear that, from the completeness of Beller et al.'s procedure, all right-maximal text substrings are visited by the procedure, since otherwise the LCP values equal to $\ell = |aX|$ inside $\mathtt{range(aX)}$ would not be generated. It follows that, in order to generate all suffix tree intervals \emph{once}, we need two extra ingredients: (i) whenever we pop from $Q_\ell$ an element $\langle\langle L,R \rangle, |W| \rangle$ corresponding to a string $W = Xc$, we also need the range of $X$, and (ii) we need to quickly check if a given range $\mathtt{range(aX)}$ of a right-maximal substring $aX$ has already been output. Point (ii) is necessary since, using only the above procedure (augmented with point (i)), $\mathtt{range(aX)}$ will be output for each of its right-extensions (except the lexicographically largest, which does not cause the generation of an LCP pair).

Remember that, in order to keep space usage under control (i.e. $O(n)$ bits), we represent $Q_\ell$ as a standard queue of pairs $\langle \mathtt{range(W)}, |W| \rangle$ if and only if $|Q_\ell| < n/\log n$. For now, let us assume that the queue size does not exceed this quantity (the other case will be considered later). 
In this case, to implement point (i) we simply augment queue pairs as $\langle \mathtt{range(W)}, \mathtt{range(X)}, |W| \rangle$, where $W=Xc$ for some $c\in\Sigma$. When left-extending $W$ with a character $a$, we also left-extend $X$ with $a$, obtaining $\mathtt{range(aX)}$. Let $\mathtt{range(aW)} = \langle L_a, R_a\rangle$. At this point, if $U[R_a+1]=0$ we do the following:
\begin{enumerate}
    \item we set $U[R_a+1]\leftarrow 1$,
    \item we push $\langle \mathtt{range(aW)}, \mathtt{range(aX)}, |W|+1 \rangle$ in $Q_{\ell+1}$, and
    \item if $\mathtt{range(aX)}$ has not already been generated, we output $\mathtt{range(aX)}$.
\end{enumerate}
Note that steps (1) and (2) correspond to Beller et al.'s procedure. The test in step (3) (that is, point (ii) above) can be implemented as follows. Note that a suffix array range $\mathtt{range(aX)} = \langle L,R\rangle$ can be identified unambiguously by the two integers $L$ and $|aX| = \ell$. Note also that we generate suffix tree intervals in increasing order of string depth (i.e. when popping elements from $Q_\ell$, we output suffix array intervals of string depth $\ell$). It follows that we can keep a bitvector $GEN_\ell$ of length $n$ recording in $GEN_\ell[i]$ whether or not the suffix array interval of the string of length $\ell$ whose first coordinate is $i$ has already been output. Each time we change the value of a bit $GEN_\ell[i]$ from 0 to 1, we also push $i$ into a stack $SET_\ell$. 
Let us assume for now that also $SET_\ell$'s size does not exceed $n/\log n$ (later we will consider a different representation for the other case). Then, also the bit-size of $SET_\ell$ will never exceed $O(n)$ bits. After $Q_\ell$ has been emptied, for each $i\in SET_\ell$ we set $GEN_\ell[i] \leftarrow 0$. This makes all  $GEN_\ell$'s entries equal to 0, and we can thus re-use its space for $GEN_{\ell+1}$ at the next stage (i.e. when popping elements from $Q_{\ell+1}$).

Now, let us consider the case  $|Q_\ell| \geq n/\log n$. The key observation is that $Q_\ell$ exceeds this value for at most $O(\log n)$ values of $\ell$, therefore we can afford spending extra $O(n/\log n)$ time to process each of these queues. As seen in Section \ref{sec:beller}, whenever $Q_\ell$'s size exceeds $n/\log n$ (while pushing elements in it) we switch to a different queue representation using packed bitvectors. 
Point (i) can be solved by storing two additional bitvectors as follows. Suppose we are about to push the triple $\langle \mathtt{range(W)}, \mathtt{range(X)}, |W| \rangle$ in $Q_\ell$, where $W=Xc$ for some $c\in\Sigma$. The solution seen in Section  \ref{sec:beller} consisted in marking, in two packed bitvectors $\mathtt{open[1,n]}$ and $\mathtt{close[1,n]}$, the start and end points of $\mathtt{range(W)}$. Now, we just use two additional packed bitvectors $\mathtt{\overline{open}[1,n]}$ and $\mathtt{\overline{close}[1,n]}$ to also mark the start and end points of $\mathtt{range(X)}$. As seen in  Section \ref{sec:beller}, intervals are extracted from $Q_\ell$ by scanning $\mathtt{open[1,n]}$ and $\mathtt{close[1,n]}$ in $O(n/\log n + |Q_\ell|)$ time (exploiting word-parallelism). Note that $W$ is a right-extension of $X$, therefore $\mathtt{range(W)}$ is contained in $\mathtt{range(X)}$. It follows that we can scan in parallel the bitvectors $\mathtt{open[1,n]}$, $\mathtt{close[1,n]}$,  $\mathtt{\overline{open}[1,n]}$, and $\mathtt{\overline{close}[1,n]}$ and retrieve, for each $\mathtt{range(W)}$ extracted from the former two bitvectors, the (unique in the queue) interval $\mathtt{range(X)}$ enclosing $\mathtt{range(W)}$ (using the latter two bitvectors). 
More formally, whenever finding a bit set at $\mathtt{\overline{open}[i]}$, we search  $\mathtt{\overline{close}[i,n]}$ to find the next bit set. Let us call $j$ the position containing such bit set. Then, we similarly scan $\mathtt{open[i,j]}$ and $\mathtt{close[i,j]}$ to generate all intervals $\langle l,r \rangle$ enclosed by $\langle i,j \rangle$, and for each of them generate the triple $\langle \langle l,r \rangle, \langle i,j \rangle, \ell \rangle$.
Again, exploiting word-parallelism the process takes $O(n/\log n + |Q_\ell|)$ time to extract all triples $\langle \mathtt{range(W)}, \mathtt{range(X)}, |W| \rangle$ from $Q_\ell$.

A similar solution can be used to solve point (ii) for large $SET_\ell$. Whenever $SET_\ell$ exceeds size $n/\log n$, we simply empty it and just use bitvector $GEN_\ell$. This time, however, this bitvector is packed in $O(n/\log n)$ words. It can therefore be erased (i.e. setting all its entries to 0) in $O(n/\log n)$ time, and we do not need to use the stack $SET_\ell$ at all. Since (a) we insert an element in some $SET_\ell$ only when outputting a suffix tree range and (b) in total we output $O(n)$ such ranges, $SET_\ell$ can exceed size $n/\log n$ for at most $O(\log n)$ values of $\ell$. We conclude that also the cost of creating and processing all $GEN_\ell$ and $SET_\ell$ amortizes to $O(n)$.

To sum up, the overall procedure runs in $O(n\log\sigma)$ time and uses $O(n)$ bits of space. By combining it with Belazzougui's procedure as seen above (i.e. choosing the right procedure according to the alphabet's size), we obtain: 

\begin{lemma}\label{lem:ST intervals}
    Given a wavelet tree representing the Burrows-Wheeler transform of a text $T$ of length $n$ on alphabet $[1,\sigma]$, in $O(n\log\sigma)$ time and $o(n\log\sigma)$ bits of working space we can enumerate the suffix array intervals corresponding to all right maximal text's substrings. 
\end{lemma}

\section{Building the PLCP Bitvector}\label{sec:PLCP}

The PLCP array is defined as $PLCP[i] = LCP[ISA[i]]$, and can thus be used to retrieve LCP values as $LCP[i] = PLCP[SA[i]]$ (note that this requires accessing the suffix array). Kasai et al. showed in~\cite{10.1007/3-540-48194-X_17} that PLCP is almost increasing: $PLCP[i+1] \geq PLCP[i]-1$. This allows representing it in small space as follows. Let $\mathtt{plcp[1,2n]}$ denote the bitvector having a bit set at each position $PLCP[i]+2i$, for $i=1, \dots, n$ (and 0 in all other positions). Since $PLCP[i+1] \geq PLCP[i]-1$, the quantity $PLCP[i]+2i$ is different for each $i$. By definition, $PLCP[i]$ can be written as $j-2i$, where $j$ is the position of the $i$-th bit set in $\mathtt{plcp}$; this shows that each PLCP entry can be retrieved in constant time using the bitvector $\mathtt{plcp}$, augmented to support constant-time \emph{select} queries. 

We now show how to build the $\mathtt{plcp}$ bitvector in small space using the LCP enumeration procedure of Section \ref{sec:LCP}. Our procedure relies on the concept of \emph{irreducible LCP values}:

\begin{definition}\label{def:irreducible}
$LCP[i]$ is said to be \emph{irreducible} if and only if either $i=0$ or $BWT[i] \neq BWT[i-1]$ hold.
\end{definition}

We call \emph{reducible} a non-irreducible LCP value. 
We extend the above definition to PLCP values, saying that $PLCP[i]$ is irreducible if and only if $LCP[ISA[i]]$ is irreducible.
The following Lemma, shown in~\cite{10.1007/978-3-540-27810-8_32}, is easy to prove (see also ~\cite[Lem. 4]{10.1007/978-3-642-02441-2_17}): 

\begin{lemma}[\cite{10.1007/978-3-540-27810-8_32}, Lem. 1]\label{lem:reducible}
    If $PLCP[i]$ is reducible, then $PLCP[i] = PLCP[i - 1] - 1$.
\end{lemma}

We also make use of the following Theorem from K{\"a}rkk{\"a}inen et al.~\cite{10.1007/978-3-642-02441-2_17}:

\begin{theorem}[\cite{10.1007/978-3-642-02441-2_17}, Thm. 1]\label{lem:sum irreducible}
    The sum of all irreducible lcp values is at most $2n\log n$.
\end{theorem}

Our strategy is as follows. We divide $BWT[1,n]$ in $\lceil n/B \rceil$ blocks $BWT[(i-1)\cdot B+1,i\cdot B]$, $i=1, \dots, \lceil n/B \rceil$ of size $B$ (assume for simplicity that $B$ divides $n$). For each block $i=1, \dots, \lceil n/B \rceil$, we use Lemma \ref{thm:LCP collection} to enumerate all pairs $(j,LCP[j])$. Whenever we generate a pair $(j,LCP[j])$ such that 
(i) $j$ falls in the current block's range $[(i-1)\cdot B+1,i\cdot B]$, (ii) $LCP[j]> \log^3 n$, and (iii) $LCP[j]$ is irreducible (this can be checked easily using Definition \ref{def:irreducible}), we store $(j,LCP[j])$ in a temporary array $\mathtt{LARGE\_LCP}$ (note: each such pair requires $O(\log n)$ bits to be stored). By Theorem \ref{lem:sum irreducible}, there cannot be more than $2n/\log^2 n$ irreducible LCP values being larger than $\log^3 n$, that is, $\mathtt{LARGE\_LCP}$ will never contain more than $2n/\log^2 n$ values and its bit-size will never exceed $O(n/\log n) = o(n)$ bits. 
We also mark all such relative positions $j-(i-1)\cdot B$ in a bitvector of length $B$  with rank support and radix-sort $\mathtt{LARGE\_LCP}$ in $O(B)$ time to guarantee constant-time access to $LCP[j]$ whenever conditions (i-iii) hold true for index $j$.
On the other hand, if (i) $j$ falls in the current block's range $[(i-1)\cdot B+1,i\cdot B]$, (ii) $LCP[j] \leq \log^3 n$, and (iii) $LCP[j]$ is irreducible then we can store $LCP[j]$ in another temporary vector $\mathtt{SMALL\_LCP[1,B]}$ as follows: $\mathtt{SMALL\_LCP}[j-(i-1)\cdot B] \leftarrow LCP[j]$ (at the beginning, the vector is initialized with undefined values). By condition (ii), $\mathtt{SMALL\_LCP}$ can be stored in $O(B\log\log n)$ bits. Using $\mathtt{LARGE\_LCP}$ and $\mathtt{SMALL\_LCP}$, we can access in constant time all irreducible values $LCP[j]$ whenever $j$ falls in the current block $[(i-1)\cdot B+1,i\cdot B]$. At this point, we enumerate all pairs $(i,ISA[i])$ in text order (i.e. for $i=1, \dots, n$) using the FL function on the BWT. Whenever one of those pairs $(i,ISA[i]) = (i,j)$ is such that (i) $j$ falls in the current block's range $[(i-1)\cdot B+1,i\cdot B]$ and (ii) $LCP[j]$ is irreducible, we retrieve $LCP[j]$ in constant time as seen above and we set $\mathtt{plcp[2i+LCP[j]]} \leftarrow 1$; the correctness of this assignment follows from the fact that $j=ISA[i]$, thus $LCP[j] = PLCP[i]$. Using Lemma \ref{lem:reducible}, we can moreover compute the reducible PLCP values that follow $PLCP[i]$ in text order (up to the next irreducible value), and set the corresponding bits in $\mathtt{plcp}$. After repeating the above procedure for all blocks $BWT[(i-1)\cdot B+1,i\cdot B]$, $i=1, \dots, \lceil n/B \rceil$, we terminate the computation of bitvector $\mathtt{plcp}$. For each block, we spend $O(n\log\sigma)$ time (one application of Lemma \ref{thm:LCP collection} and one BWT navigation to generate all pairs $(i,ISA[i])$).
We also spend $O(n/\log^2 n)$ time to allocate the instances of $\mathtt{LARGE\_LCP}$ across all blocks.
Overall, we spend $O((n^2/B)\log\sigma + n\log\sigma)$ time across all blocks. 
The space used is $o(n) + O(B\cdot \log\log n)$ bits on top of the BWT. By setting $B = (\epsilon\cdot n\log\sigma)/\log\log n$ we obtain our result: 

\begin{lemma}\label{lem:PLCP}
    Given a wavelet tree for the Burrows-Wheeler transform of a text $T$ of length $n$ on alphabet $[1,\sigma]$, for any parameter $0<\epsilon \leq 1$ we can build the PLCP bitvector in $O(n(\log\sigma + \epsilon^{-1}\log\log n))$ time and $\epsilon \cdot n\log\sigma + o(n)$ bits of working space on top of the input BWT and the optput. 
\end{lemma}

\section{Building the Suffix Tree Topology}\label{sec:ST topology}

In order to build the suffix tree topology we use a strategy analogous to the one proposed by Belazzougui~\cite{Belazzougui:2014:LTC:2591796.2591885}. The main observation is that, given a procedure that enumerates suffix tree intervals, for each interval $[l,r]$ we can increment a counter $\mathtt{Open[l]}$ and a counter $\mathtt{Close[r]}$, where $\mathtt{Open}$  and $\mathtt{Close}$ are integer vectors of length $n$. Then, the BPS representation of the suffix tree topology can be built by scanning left-to right the two arrays and, for each $i=1, \dots, n$, append $\mathtt{Open[i]}$ open parentheses followed by $\mathtt{Close[i]}$ close parentheses to the BPS representation. The main drawback of this solution is that it takes too much space: $2n\log n$ bits to store the two arrays. Belazzougui solves this problem by noticing that the sum of all the values in the two arrays is the length of the final BPS representation, that is, at most $4n$. This makes it possible to represent the arrays in just $O(n)$ bits of space by representing (the few) large counters in plain form and (the many) small counters using delta encoding (while still supporting updates in constant time). 

Our goal in this section is to reduce the working space from $O(n)$ to a (small) fraction of $n\log \sigma$. A first idea could be to iterate Belazzougui's strategy on chunks of the interval $[1,n]$.  Unfortunately, this does not immediately give the correct solution as a chunk could still account for up to $\Theta(n)$ parentheses, no matter what the length of the chunk is; as a result, Belazzougui's representation could still take $O(n)$ bits of space for a chunk (when using large enough chunks to keep the running time under control as seen in the previous section). We use a solution analogous to the one discussed in the previous section. This solution corresponds to the first part of Belazzougui's strategy (in particular, we will store small counters in plain form instead of using delta encoding). We divide $BWT[1,n]$ in $\lceil n/B \rceil$ blocks $BWT[(i-1)\cdot B+1,i\cdot B]$, $i=1, \dots, \lceil n/B \rceil$ of size $B$ (assume for simplicity that $B$ divides $n$). For each block $i=1, \dots, \lceil n/B \rceil$, we use Lemma \ref{lem:ST intervals} to enumerate all suffix tree intervals $[l,r]$. We keep two arrays $\mathtt{Open[1,B]}$ and $\mathtt{Close[1,B]}$ storing integers of $2\log\log n$ bits each. Whenever the beginning $l$ of a suffix tree interval $[l,r]$ falls inside the current block $[(i-1)\cdot B+1,i\cdot B]$, we increment $\mathtt{Open[l- (i-1)\cdot B]}$ (the description is analogous for index $r$ and array $\mathtt{Close}$). If $\mathtt{Open[l- (i-1)\cdot B]}$ reaches the maximum value $2^{2\log\log n-1}$, we no longer increment it. Adopting Belazzougui's terminology, we call such a bucket ``saturated''. 
After having generated all suffix tree intervals, 
let $k$ be the number of saturated counters. 
We allocate a vector $\mathtt{LARGE\_COUNTERS}$ storing $k$ integers of $\log n + 2$ bits each (enough to store the value $4n$, i.e. an upper-bound to the value that a counter can reach). We also allocate a bitvector of length $B$ marking saturated counters, and process it to support constant-time rank queries. This allows us to obtain in constant time the location in $\mathtt{LARGE\_COUNTERS}$ corresponding to any saturated counter in the block. We generate all suffix tree intervals for a second time using again Lemma \ref{lem:ST intervals}, this time incrementing (in $\mathtt{LARGE\_COUNTERS}$) only locations corresponding to saturated counters. Since the BPS sequence has length at most $4n$ and a counter saturates when it reaches value $\Theta(\log^2 n)$, we have that $k = O(n/\log^2n)$ and thus $\mathtt{LARGE\_COUNTERS}$ takes at most $O(n/\log n) = o(n)$ bits to be stored.
The rest of the analysis is identical to the algorithm described in the previous section. For each block, we spend $O(n\log\sigma)$ time (two applications of Lemma \ref{lem:ST intervals}). We also spend $O(n/\log^2 n)$ time to allocate the instances of $\mathtt{LARGE\_COUNTERS}$ across all blocks. 
Overall, we spend $O((n^2/B)\log\sigma + n\log\sigma)$ time across all blocks. 
The space used is $o(n) + O(B\cdot \log\log n)$ bits on top of the BWT. By setting $B = (\epsilon\cdot n\log\sigma)/\log\log n$ we obtain: 

\begin{lemma}\label{lem:BPS}
    Given a wavelet tree for the Burrows-Wheeler transform of a text $T$ of length $n$ on alphabet $[1,\sigma]$, for any parameter $0<\epsilon \leq 1$ we can build the BPS representation of the suffix tree topology in $O(n(\log\sigma + \epsilon^{-1}\log\log n))$ time and $\epsilon \cdot n\log\sigma + o(n)$ bits of working space on top of the input BWT and the optput. 
\end{lemma}

To conclude, we note that our procedures can be immediately used to build space-efficiently the compressed suffix tree described by Sadakane~\cite{Sadakane:2007:CST:1326296.1326297} starting from the BWT.
The only missing ingredients are (i) to augment the BWT with a suffix array sample in order to turn it into a CSA, and (ii) to pre-process the PLCP and BPS sequences to support fast queries (\emph{select} on the PLCP and navigational queries on the BPS). Step (i) can be easily performed in $O(n\log\sigma)$ time and $n+o(n)$ bits of working space with a folklore solution that iteratively applies function LF to navigate all BWT's positions and collect one suffix array sample every $O(\log^{1+\delta} n/\log\sigma)$ text positions, for any fixed $\delta>0$ (using a succinct bitvector to mark sampled positions). The resulting CSA takes $n\log\sigma + o(n\log\sigma)$ bits of space and allows computing any $SA[i]$ in $O(\log^{1+\delta} n)$ time.
Step (ii) can be performed in $O(n)$ time and $o(n)$ bits of working space using textbook solutions (see \cite{Navarro:2016:CDS:3092586}).
Combining this with Lemmas \ref{thm:BWT->WT}, \ref{lem:PLCP}, and \ref{lem:BPS}, we obtain: 

\begin{theorem}\label{lem:BPS}
    Given the word-packed BWT of a text $T$ of length $n$ on alphabet $[1,\sigma]$, for any parameter $0<\epsilon \leq 1$ we can replace it
    in $O(n(\log\sigma + \epsilon^{-1}\log\log n))$ time and $\epsilon \cdot n\log\sigma + o(n)$ bits of working space
    with a compressed suffix tree taking $n\log\sigma + 6n + o(n\log\sigma)$ bits of space and supporting all operations in $O(\tt{polylog}\ n)$ time.  
\end{theorem}

\section{Merging BWTs in Small Space}\label{sec:algo2}

In this section we use our space-efficient BWT-navigation strategies to tackle an additional problem: to merge the BWTs of two string collections. 
In~\cite{Belazzougui:2014:LTC:2591796.2591885,belazzougui2016linear}, Belazzougui et al. show that Theorem \ref{th:Belazzougui} can be adapted to merge the BWTs of two texts $T_1, T_2$ and obtain the BWT of the collection $\{T_1, T_2\}$ in $O(nk)$ time and $n\log\sigma(1+1/k) + 11n + o(n)$ bits of working space for any $k \geq 1$~\cite[Thm. 7]{belazzougui2016linear}. We show that our strategy enables a more space-efficient algorithm for the task of merging BWTs of collections.
The following theorem, whose proof is reported later in this section, merges two BWTs by computing the binary DA of their union. After that, the merged BWT can be streamed to external memory (the DA tells how to interleave characters from the input BWTs) and does not take additional space in internal memory. Similarly to what we did in the proof of Theorem \ref{thm:LCP collection succinct}, this time we re-use the space of the Document Array to accommodate the extra $n$ bits needed to replace the BWTs of the two collections with their wavelet matrices. This is the main result of this section: 

\begin{theorem}\label{th:merge}
	Given the Burrows-Wheeler Transforms of two collections $\bigS_1$ and $\bigS_2$ of total length $n$ on alphabet $[1,\sigma]$, we can compute the Document Array 
	of $\bigS_1 \cup \bigS_2$ in $O(n\log\sigma)$ time using $o(n\log\sigma)$ bits of working space on top of the input BWTs and the output DA.
\end{theorem}

We also briefly discuss how to extend Theorem \ref{th:merge} to build the LCP array of the merged collection.
In Section \ref{sec:experiments} we present an  implementation of our algorithms and an experimental comparison  with \eGap~\cite{egidi2017lightweight}, the state-of-the-art tool designed for the same task of merging BWTs while inducing the LCP of their union.

The procedure of Algorithm \ref{alg:fill leaves} can be extended to merge BWTs of two collections $\bigS_1$, $\bigS_2$ using $o(n\log\sigma)$ bits of working space on top of the input BWTs and output Document Array (here, $n$ is the cumulative length of the two BWTs). The idea is to simulate a navigation of the \emph{leaves} of the generalized suffix tree of $\bigS_1 \cup \bigS_2$ (note: for us, a collection is an ordered multi-set of strings). 
Our procedure differs from that described in~\cite[Thm. 7]{belazzougui2016linear} in two ways. First, they navigate a subset of the suffix tree \emph{nodes} (so-called \emph{impure} nodes, i.e. the roots of subtrees containing suffixes from distinct strings), whereas we navigate leaves. Second, their visit is implemented by following Weiner links. This forces them to represent the nodes with the ``heavy'' representation $\mathtt{repr}$ of Section \ref{sec:belazzougui}, which is not efficient on large alphabets. On the contrary, leaves can be represented simply as ranges and allow for a more space-efficient queue/stack representation. 

We represent each leaf by a pair of intervals, respectively on $BWT(\bigS_1)$ and $BWT(\bigS_2)$, of strings of the form $W\#$. Note that: (i) the suffix array of $\bigS_1 \cup \bigS_2$ is covered by the non-overlapping intervals of strings of the form $W\#$, and (ii) for each such string $W\#$, the interval $\mathtt{range(W\#)} = \langle L,R \rangle$ in $GSA(\bigS_1 \cup \bigS_2)$ can be partitioned as $\langle L, M \rangle \cdot \langle M+1, R\rangle$, where $\langle L,M\rangle$ contains only suffixes from  $\bigS_1$ and $\langle M+1,R \rangle$ contains only suffixes from  $\bigS_2$ (one of these two intervals could be empty).	
It follows that we can navigate in parallel the leaves of the suffix trees of $\bigS_1$ and $\bigS_2$ (using again a stack or a queue containing pairs of intervals on the two BWTs), and fill the Document Array $DA[1,n]$, an array that will tell us whether the $i$-th entry of $BWT(\bigS_1 \cup \bigS_2)$ comes from $BWT(\bigS_1)$ ($DA[i] = 0$) or $BWT(\bigS_2)$ ($DA[i] = 1$). To do this, let $\langle L_1, R_1\rangle$ and $\langle L_2, R_2\rangle$ be the ranges on the suffix arrays of $\bigS_1$ and $\bigS_2$, respectively, of a  suffix $W\#$ of some string in the collections. 
Note that one of the two intervals could be empty: $R_j<L_j$. In this case, we still require that $L_j-1$ is the number of suffixes in $\bigS_j$ that are smaller than $W\#$.
Then, in the collection $\bigS_1 \cup \bigS_2$ there are $L_1 + L_2 - 2$ suffixes smaller than $W\#$, and $R_1 + R_2$ suffixes smaller than or equal to $W\#$. It follows that the range of $W\#$ in the suffix array of $\bigS_1 \cup \bigS_2$ is $\langle L_1+L_2-1, R_1+R_2\rangle$, where the first $R_1-L_1+1$ entries correspond to suffixes of strings from $\bigS_1$. Then, we set $DA[L_1+L_2-1, L_2 + R_1-1] \leftarrow 0$ and $DA[L_2 + R_1,R_1+R_2] \leftarrow 1$. 
The procedure starts from the pair of intervals corresponding to the ranges of the string ``$\#$'' in the two BWTs, and proceeds recursively by left-extending the current pair of ranges $\langle L_1, R_1\rangle$, $\langle L_2, R_2\rangle$ with the symbols in $\mathtt{BWT_1.rangeDistinct(L_1,R_1)} \cup \mathtt{BWT_2.rangeDistinct(L_2,R_2)}$.
The detailed procedure is reported in Algorithm \ref{alg:merge}. The leaf visit is implemented, again, using a stack or a queue; this time however, these containers are filled with pairs of intervals $\langle L_1, R_1\rangle$, $\langle L_2, R_2\rangle$. 
We implement the stack simply as a vector of quadruples $\langle L_1, R_1, L_2, R_2\rangle$. As far as the queue is concerned, some care needs to be taken when representing the pairs of ranges using bitvectors as seen in Section \ref{sec:beller} with Beller et al.'s representation. 
Recall that, at any time, the queue can be partitioned in two sub-sequences associated with LCP values $\ell$ and $\ell+1$ (we pop from the former, and push in the latter).
This time, we represent each of these two subsequences as a vector of quadruples (pairs of ranges on the two BWTs) as long as the number of quadruples in the sequence does not exceed $n/\log n$. When there are more quadruples than this threshold, we switch to a bitvector representation defined as follows. 
Let $|BWT(\bigS_1)|=n_1$, $|BWT(\bigS_2)|=n_2$, and $|BWT(\bigS_1\cup \bigS_2)| = n = n_1+n_2$.
We keep two bitvectors $\mathtt{Open[1,n]}$ and $\mathtt{Close[1,n]}$ storing opening and closing parentheses of intervals in $BWT(\bigS_1\cup \bigS_2)$. We moreover keep two bitvectors $\mathtt{NonEmpty_1[1,n]}$ and $\mathtt{NonEmpty_2[1,n]}$ keeping track, for each $i$ such that $\mathtt{Open[i]=1}$, of whether the interval starting in $BWT(\bigS_1\cup \bigS_2)[i]$ contains suffixes of reads coming from $\bigS_1$ and $\bigS_2$, respectively. Finally, we keep four bitvectors $\mathtt{Open_j[1,n_j]}$ and $\mathtt{Close_j[1,n_j]}$, for $j=1,2$, storing non-empty intervals on $BWT(\bigS_1)$ and $BWT(\bigS_2)$, respectively. To insert a pair of intervals $\langle L_1, R_1\rangle,\ \langle L_2, R_2\rangle$ in the queue, let $\langle L,R \rangle = \langle L_1+L_2-1, R_1+R_2\rangle$. We set $\mathtt{Open[L]} \leftarrow 1$ and  $\mathtt{Close[R]} \leftarrow 1$. Then, for $j=1,2$, we set $\mathtt{NonEmpty_j[L]} \leftarrow 1$, $\mathtt{Open_j[L_j]} \leftarrow 1$ and $\mathtt{Close_j[R_j]} \leftarrow 1$ if and only if $R_j\geq L_j$. 
This queue representation takes $O(n)$ bits. 
By construction, for each bit set in $\mathtt{Open}$ at position $i$, there is a corresponding bit set in 
$\mathtt{Open_j}$ if and only if $\mathtt{NonEmpty_j[i]} = 1$ (moreover, corresponding bits set appear in the same order in $\mathtt{Open}$ and
$\mathtt{Open_j}$). It follows that a left-to-right scan of these bitvectors is sufficient to identify corresponding intervals on $BWT(\bigS_1\cup \bigS_2)$, $BWT(\bigS_1)$, and $BWT(\bigS_2)$.
By packing the bits of the bitvectors in words of $\Theta(\log n)$ bits, the $t$ pairs of intervals contained in the queue can be extracted in $O(t+ n/\log n)$ time (as described in~\cite{beller2013computing}) by scanning in parallel the bitvectors forming the queue. Particular care needs to be taken only when we find the beginning of an interval $\mathtt{Open[L]=1}$ with $\mathtt{NonEmpty_1[L]} = 0$ (the case $\mathtt{NonEmpty_2[L]} = 0$ is symmetric). Let $L_2$ be the beginning of the corresponding non-empty interval on $BWT(\bigS_2)$. Even though we are not storing $L_1$ (because we only store nonempty intervals), we can retrieve this value as $L_1=L-L_2+1$. Then, the empty interval on $BWT(\bigS_1)$ is $\langle L_1, L_1-1\rangle$.

The same arguments used in the previous section show that the algorithm runs in $O(n\log\sigma)$ time and uses $o(n\log\sigma)$ bits of space on top of the input BWTs and output Document Array. This proves Theorem \ref{th:merge}.
To conclude, we note that the algorithm can be easily extended to compute the LCP array of the merged collection while merging the BWTs. This requires adapting Algorithm \ref{alg:fill nodes} to work on pairs of suffix tree nodes (as we did in Algorithm \ref{alg:merge} with pairs of leaves). Results on an implementation of the extended algorithm are discussed in the next section. 
From the practical point of view, note that it is more advantageous to induce the LCP of the merged collection while merging the BWTs (rather than first merging and then inducing the LCP using the algorithm of the previous section), since leaf-type LCP values can be induced directly while computing the document array.

\begin{algorithm}
  \begin{algorithmic}[1]
  
   \If{$\sigma > n/\log^3n$}
			
      \State $\mathtt P \leftarrow \mathtt{new\_queue()}$\Comment{Initialize new queue of interval pairs}\label{new queue3}
   \Else
   	 \State $\mathtt P \leftarrow \mathtt{new\_stack()}$\Comment{Initialize new stack of interval pairs}\label{new stack3}
   \EndIf

   \State $\mathtt P\mathtt{.push(BWT_1.range(\#),BWT_2.range(\#))}$\Comment{Push SA-ranges of terminator}\label{push5}
   
   \While{$\mathtt{\mathbf{not}\ P.empty()}$}\label{while3}
   
   \State $\langle L_1,R_1, L_2, R_2 \rangle \leftarrow \mathtt{P.pop()}$\Comment{Pop highest-priority element}\label{pop3}

  \For{$i=L_1+L_2-1\dots L_2+R_1-1$}
	\State $\mathtt{DA}[i] \leftarrow 0$\Comment{Suffixes from $\bigS_1$}\label{DA1}
  \EndFor

  \For{$i=L_2+R_1\dots R_1+R_2$}
    \State $\mathtt{DA}[i] \leftarrow 1$\Comment{Suffixes from $\bigS_2$}\label{DA2}
 \EndFor

  \If{$\sigma > n/\log^3n$}
     \State$\mathtt{P.push(getIntervals(L_1, R_1, L_2, R_2, BWT_1, BWT_2))}$\Comment{New intervals}\label{push8}
  \Else
     \State $c_1^1, \dots, c_{k_1}^1 \leftarrow \mathtt{BWT_1.rangeDistinct(L_1,R_1)}$\label{range distinct3}
     \State $c_1^2, \dots, c_{k_2}^2 \leftarrow \mathtt{BWT_2.rangeDistinct(L_2,R_2)}$\label{range distinct4}
     \State $\{c_1\dots c_k\} \leftarrow \{c_1^1, \dots, c_{k_1}^1\} \cup \{c_1^2, \dots, c_{k_2}^2\}$\label{range distinct5}
    
    \For{$i=1\dots k$}
      \State $\langle L_1^i, R_1^i\rangle \leftarrow \mathtt{BWT_1.bwsearch(\langle L_1, R_1\rangle, c_i)}$\Comment{Backward search step}\label{BWS2}
    \EndFor

   \For{$i=1\dots k$}
     \State $\langle L_2^i, R_2^i\rangle \leftarrow \mathtt{BWT_2.bwsearch(\langle L_2, R_2\rangle, c_i)}$\Comment{Backward search step}\label{BWS3}
     
    \EndFor
     
     \State $\langle \hat L_1^i, \hat R_1^i, \hat L_2^i, \hat R_2^i, \rangle_{i=1, \dots, k} \leftarrow \mathtt{sort}(\langle L_1^i, R_1^i, L_2^i, R_2^i, \rangle_{i=1, \dots, k})$\label{sort3}

     \For{$i=k\dots 1$}

      \State $\mathtt{P.push}(\hat L_1^i, \hat R_1^i, \hat L_2^i, \hat R_2^i)$\Comment{Push in order of decreasing length}\label{push6}

   \EndFor

 \EndIf

 \EndWhile

 \caption{$\mathtt{Merge(BWT_1,BWT_2, DA)}$}\label{alg:merge}
\end{algorithmic}
\end{algorithm}

Note that Algorithm \ref{alg:merge} is similar to Algorithm \ref{alg:fill leaves}, except that now we manipulate pairs of intervals. In Line \ref{sort3}, we sort quadruples according to the length $R_1^i + R_2^i - (L_1^i + L_2^i) +2$ of the combined interval on $BWT(\bigS_1\cup \bigS_2)$. Finally, note that Backward search can be performed correctly also when the input interval is empty: $\mathtt{BWT_j.bwsearch(\langle L_j, L_j-1 \rangle, c)}$, where $L_j-1$ is the number of suffixes in $\bigS_j$ smaller than some string $W$, correctly returns the pair $\langle L', R'\rangle$ such that $L'$ is the number of suffixes in $\bigS_j$ smaller than $cW$: this is true when implementing backward search with a $rank_c$ operation on position $L_j$; then, if the original interval is empty we just set $R'=L'-1$ to keep the invariant that $R'-L'+1$ is the interval's length. 

\section{Implementation and Experimental Evaluation}\label{sec:experiments}

We implemented our LCP construction and BWT merge algorithms on DNA alphabet in \repo using the language C++. 
Due to the small alphabet size, it was actually sufficient to implement our extension of Belazzougui's enumeration algorithm (and not the strategy of Beller et al., which becomes competitive only on large alphabets).
The repository features a new packed string on DNA alphabet $\Sigma_{DNA}=\{A,C,G,T,\#\}$ using 4 bits per character and able to compute the quintuple $\langle BWT.rank_c(i) \rangle_{i\in \Sigma_{DNA}}$ with just one cache miss.  This is crucial for our algorithms, since at each step we need to left-extend ranges by all characters. 
This structure divides the text in blocks of 128 characters. Each block is stored using 512 cache-aligned bits (the typical size of a cache line), divided as follows. The first 128 bits store four 32-bits counters with the partial ranks of A, C, G, and T before the block (if the string is longer than $2^{32}$ characters, we further break it into superblocks of $2^{32}$ characters; on reasonably-large inputs, the extra rank table fits in cache and does not cause additional cache misses). 
The following three blocks of 128 bits store the first, second, and third bits, respectively, of the characters' binary encodings (each character is packed in 3 bits). Using this layout, the rank of each character in the block can be computed with at most three masks, a bitwise AND (actually less, since we always compute the rank of all five characters and we re-use partial results whenever possible), and a \texttt{popcount} operation. 
We also implemented a packed string on the augmented alphabet $\Sigma_{DNA}^+=\{A,C,G,N,T,\#\}$ using $4.38$ bits per character and offering the same cache-efficiency guarantees.
In this case, a 512-bits block stores 117 characters, packed as follows.
As seen above, the first 128 bits store four 32-bits counters with the partial ranks of A, C, G, and T before the block. 
Each of the following three blocks of 128 bits is divided in a first part of 117 bits and a second part of 11 bits. The first parts store the first, second, and third bits, respectively, of the characters' binary encodings. The three parts of 11 bits, concatenated together, store the rank of N's before the block. This layout minimizes the number of bitwise operations (in particular, shifts and masks) needed to compute a parallel rank.

Several heuristics have been implemented to reduce the number of cache misses in practice. In particular, we note that in Algorithm \ref{alg:fill leaves} we can avoid backtracking when the range size becomes equal to one; the same optimization can be implemented in Algorithm \ref{alg:merge} when also computing the LCP array, since leaves of size one can be identified during navigation of internal suffix tree nodes. Overall, we observed (using a memory profiler) that in practice the combination of Algorithms \ref{alg:fill nodes}-\ref{alg:fill leaves} generates at most $1.5n$ cache misses, $n$ being the total collection's size. The extension of Algorithm \ref{alg:merge} that computes also LCP values generates twice this number of cache misses (this is expected, since the algorithm navigates two BWTs).

We now report some preliminary experiments on our algorithms: \induceLCP (Algorithms \ref{alg:fill nodes}-\ref{alg:fill leaves}) and \mergeBwtLCP (Algorithm \ref{alg:merge}, extended to compute also the LCP array). 
All tests were done on a DELL PowerEdge R630 machine, used in non exclusive mode.
Our platform is a $24$-core machine with Intel(R) Xeon(R) CPU E5-2620 v3 at $2.40$ GHz, with $128$ GiB of shared memory and 1TB of SSD. The system is Ubuntu 14.04.2 LTS. The code was compiled using gcc 8.1.0 with flags \texttt{-Ofast} \texttt{-fstrict-aliasing}.

\begin{table}[t]
	\centering
	{\scriptsize
		\begin{tabular}{|@{\ }l@{\ }|@{\ }c@{\ }|@{\ }c@{\ }|@{\ }c@{\ }|@{\ }c@{\ }|c@{\ }|c@{\ }|}
			\hline
			Name        & Size & $\sigma$          & N. of  & Max read      & Bytes for\\
			& GiB   &              & reads  & length   & lcp values\\
			\hline
			NA12891.8           & 8.16      & 5               & 85,899,345  & 100 & 1   \\
			\hline
			shortreads         & 8.0     & 6               & 85,899,345  & 100      & 1  \\
			\hline
			pacbio                & 8.0      & 6            & 942,248      & 71,561  & 4   \\
			\hline
			pacbio.1000       & 8.0      & 6              & 8,589,934    & 1000    & 2\\
			\hline
			NA12891.24        &   23.75    & 6               & 250,000,000  & 100 & 1   \\
			\hline
			NA12878.24           &  23.75    & 6               & 250,000,000  & 100 & 1   \\
			\hline
		\end{tabular}
	}
	\caption{Datasets used in our experiments. Size accounts only for the alphabet's characters. The alphabet's size $\sigma$ includes the terminator.}
	\label{tableDataset}
\end{table}

\begin{table}[t]
	\centering
	{\scriptsize
		\begin{tabular}{|c|c|c|c|c|c|c|c|c|}
			\hline
			& \multicolumn{2}{c|}{Preprocessing}     & \multicolumn{2}{c|}{\eGap}    & \multicolumn{2}{c|}{\mergeBwtLCP}    \\ 
			\hline
			Name            &  Wall Clock &   RAM   &   Wall Clock    & RAM                & Wall Clock   & RAM      \\ 
			&  (h:mm:ss)  &    (GiB) &   (h:mm:ss)    &  (GiB)              &  (h:mm:ss)   &  (GiB)     \\
			\hline
			NA12891.8    & 1:15:57 & 2.84 & \multirow{2}{*}{10:15:07}  &   \multirow{2}{*}{18.09 (-m 32000)}    &   \multirow{2}{*}{3:16:40}    &  \multirow{2}{*}{26.52}  \\
			\cline{1-3}
			NA12891.8.RC & 1:17:55 & 2.84 &                            &                                        &                               &             \\
			\hline
			shortreads     &  1:14:51 & 2.84 & \multirow{2}{*}{11:03:10}  &    \multirow{2}{*}{16.24  (-m 29000)}  &   \multirow{2}{*}{3:36:21}    &  \multirow{2}{*}{26.75}   \\ 
			\cline{1-3}
			shortreads.RC  & 1:19:30 & 2.84 &                            &                                        &                               &             \\
			\hline
			pacbio.1000    &  2:08:56 & 31.28 & \multirow{2}{*}{5:03:01}   &   \multirow{2}{*}{21.23 (-m 45000)}    &  \multirow{2}{*}{4:03:07}     &  \multirow{2}{*}{42.75}   \\ 
			\cline{1-3}
			pacbio.1000.RC &  2:15:08 & 31.28 &                            &                                        &                               &             \\
			\hline
			pacbio         &  2:27:08 & 31.25 & \multirow{2}{*}{2:56:31}   &   \multirow{2}{*}{33.40 (-m 80000)}    &   \multirow{2}{*}{4:38:27}    &  \multirow{2}{*}{74.76}        \\ 
			\cline{1-3}
			pacbio.RC      &  2:19:27 & 31.25 &                            &                                        &                               &             \\
			\hline
			NA12878.24   & 4:24:27  & 7.69  & \multirow{2}{*}{31:12:28}   &   \multirow{2}{*}{47.50 (-m 84000)}    &   \multirow{2}{*}{6:41:35}    &  \multirow{2}{*}{73.48}        \\ 
			\cline{1-3}
			NA12891.24      & 4:02:42  & 7.69  &                            &                                        &                               &             \\
			\hline
		\end{tabular}
		\caption{In this experiment, we merge pairs of BWTs and induce the LCP of their union using \eGap and \mergeBwtLCP. We also show the resources used by the pre-processing step (building the BWTs) for comparison. Wall clock is the elapsed time from start to completion of the instance, while RAM (in GiB) is the peak Resident Set Size (RSS). All values were taken using the \texttt{/usr/bin/time} command. During the preprocessing step on the collections pacBio.1000 and pacBio, the available memory in MB (parameter m) of \eGap was set to 32000 MB. In the merge step this parameter was set to about to the memory used by \mergeBwtLCP.
			\eGap and \mergeBwtLCP take as input the same BWT file.}
		\label{tab:merge}
	}
\end{table}	

Table \ref{tableDataset} summarizes the datasets used in our experiments. 
``NA12891.8''\footnote{{\scriptsize \url{ftp://ftp.1000genomes.ebi.ac.uk/vol1/ftp/phase3/data/NA12891/sequence_read/SRR622458_1.filt.fastq.gz}}} contains Human DNA reads on the alphabet $\Sigma_{DNA}$ 
where we have removed reads containing the nucleotide $N$.
``shortreads'' contains Human DNA short reads on the extended alphabet $\Sigma_{DNA}^+$. 
``pacbio'' contains PacBio RS II reads from the species \emph{Triticum aestivum} (wheat).
``pacbio.1000'' are the strings from ``pacbio'' trimmed to length 1,000. 
All the above datasets except the first have been download from \url{https://github.com/felipelouza/egap/tree/master/dataset}.
To conclude, we added two collections, ``NA12891.24'' and ``NA12878.24'' obtained by taking the first $250,000,000$ reads from individuals NA12878\footnote{{\scriptsize \url{ftp://ftp.1000genomes.ebi.ac.uk/vol1/ftp/phase3/data/NA12878/sequence_read/SRR622457_1.filt.fastq.gz}}} and NA12891. All datasets except ``NA12891.8'' are on the alphabet $\Sigma_{DNA}^+$. In Tables \ref{tab:merge} and \ref{tab:induce}, the suffix ``.RC'' added to a dataset's name indicates the reverse-complemented dataset.

We  compare our algorithms with \eGap\footnote{{\scriptsize\url{https://github.com/felipelouza/egap}}} and BCR\footnote{{\scriptsize\url{https://github.com/giovannarosone/BCR_LCP_GSA}}}, two tools designed to build the BWT and LCP of a set of DNA reads. 
Since no tools for inducing the LCP from the BWT of a set of strings are available in the literature, in Table \ref{tab:induce} we simply compare the resources used by \induceLCP with the time and space requirements of  \eGap and BCR when building the BWT.
In \cite{EgidiAMB2019}, experimental results show that BCR works better on short reads and collections with a large average LCP, while \eGap works better when the datasets contain long reads and relatively small average LCP. 
For this reason, in the preprocessing step we have used BCR for the collections containing short reads and \eGap for the other collections.
\eGap, in addition, is capable of merging two or more BWTs while inducing the LCP of their union. In this case, we can therefore directly compare the performance of \eGap with our tool \mergeBwtLCP; results are reported in Table \ref{tab:merge}.
Since the available RAM is greater than the size of the input, we have used the semi-external strategy of \eGap. 
Notice that an entirely like-for-like comparison between our tools and \eGap is not completely feasible, being \eGap a semi-external memory tool (our tools, instead, use internal memory only).
While in our tables we report RAM usage only, it is worth to notice that \eGap uses a considerable amount of disk working space. For example, the tool uses $56$GiB of disk working space when run on a $8$GiB input (in general, the disk usage is of $7n$ bytes).

\begin{table}[t]
	{\scriptsize
		\centering
		\begin{tabular}{|c|c|c|c|c|c|c|}
			\hline
			& \multicolumn{2}{c|}{Preprocessing}  & \multicolumn{2}{c|}{\induceLCP}          \\ 
			\hline
			Name         & Wall Clock    & RAM       & Wall Clock  & RAM      \\ 
			&  (h:mm:ss)    & GiB        &  (h:mm:ss)  &  (GiB)    \\
			\hline
			NA12891.8 $\cup$ NA12891.8.RC (BCR)  &  2:43:02      & 5.67   &  1:40:01 &    24.48     \\
			\hline
			shortread $\cup$ shortread.RC (BCR) & 2:47:07    &   5.67   &  2:14:41 &   24.75     \\ 
			\hline
			pacbio.1000 $\cup$ pacbio.1000.RC (\eGap -m 32000)  &     7:07:46   &  31.28  &     1:54:56    &   40.75         \\ 
			\hline
			pacbio $\cup$ pacbio.RC (\eGap -m 80000)      &  6:02:37     &    78.125       &   2:14:37   &   72.76    \\ 
			\hline
			NA12878.24 $\cup$ NA12891.24 (BCR)      & 8:26:34      &  16.63        &   6:41:35   &   73.48    \\ 
			\hline
		\end{tabular}
		\caption{In this experiment, we induced the LCP array from the BWT of a collection (each collection is the union of two collections from Table \ref{tab:merge}). We also show pre-processing requirements (i.e. building the BWT) of the better performing tool between BCR and \eGap.}
		\label{tab:induce}
	}
\end{table}

Our tools exhibit a dataset-independent linear time complexity, whereas \eGap's running time depends on the average LCP.
Table \ref{tab:induce} shows that our tool \induceLCP induces the LCP from the BWT faster than building the BWT itself. When 'N's are not present in the dataset, \induceLCP processes data at a rate of $2.92$ megabases per second and uses $0.5$ Bytes per base in RAM in addition to the LCP. When 'N's are present, the throughput decreases to $2.12$ megabases per second and the tool uses  $0.55$ Bytes per base in addition to the LCP.
As shown in Table \ref{tab:merge}, our tool \mergeBwtLCP is from $1.25$ to $4.5$ times faster than \eGap on inputs with large average LCP, but $1.6$ times slower when the average LCP is small (dataset ``pacbio''). When 'N's are not present in the dataset, \mergeBwtLCP processes data at a rate of $1.48$ megabases per second and uses $0.625$ Bytes per base in addition to the LCP. When 'N's are present, the throughput ranges from $1.03$ to $1.32$ megabases per second and the tool uses  $0.673$  Bytes per base in addition to the LCP. When only computing the merged BWT (results not shown here for space reasons), \mergeBwtLCP uses in total $0.625$/$0.673$ Bytes per base in RAM (without/with 'N's) and is about $1.2$ times faster than the version computing also the LCP.

\bibliographystyle{plain}
\bibliography{paper}

\begin{thebibliography}{10}

\bibitem{BauerCoxRosoneTCS2013}
M.J. Bauer, A.J. Cox, and G.~Rosone.
\newblock Lightweight algorithms for constructing and inverting the {BWT} of
  string collections.
\newblock {\em Theor. Comput. Sci.}, 483(0):134--148, 2013.

\bibitem{Belazzougui:2014:LTC:2591796.2591885}
D.~Belazzougui.
\newblock Linear time construction of compressed text indices in compact space.
\newblock In {\em Proceedings of the Forty-sixth Annual ACM Symposium on Theory
  of Computing}, STOC '14, pages 148--193, New York, NY, USA, 2014. ACM.

\bibitem{belazzougui2016linear}
D.~Belazzougui, F.~Cunial, J.~K{\"a}rkk{\"a}inen, and V.~M{\"a}kinen.
\newblock Linear-time string indexing and analysis in small space.
\newblock {\em arXiv preprint arXiv:1609.06378}, 2016.

\bibitem{belazzougui2014alphabet}
D.~Belazzougui and G.~Navarro.
\newblock Alphabet-independent compressed text indexing.
\newblock {\em TALG}, 10(4):23, 2014.

\bibitem{beller2013computing}
T.~Beller, S.~Gog, E.~Ohlebusch, and T.~Schnattinger.
\newblock Computing the longest common prefix array based on the
  {B}urrows--{W}heeler transform.
\newblock {\em J. Discrete Algorithms}, 18:22--31, 2013.

\bibitem{10.1007/978-3-319-02432-5_5}
T.~Beller, M.~Zwerger, S.~Gog, and E.~Ohlebusch.
\newblock {Space-Efficient Construction of the Burrows-Wheeler Transform}.
\newblock In {\em String Processing and Information Retrieval}, pages 5--16,
  Cham, 2013. Springer International Publishing.

\bibitem{Benoit2005}
D.~Benoit, E.~D. Demaine, J.~I. Munro, R.~Raman, V.~Raman, and S.~S. Rao.
\newblock Representing trees of higher degree.
\newblock {\em Algorithmica}, 43(4):275--292, 2005.

\bibitem{Bonizzoni2018}
P.~Bonizzoni, G.~Della~Vedova, S.~Nicosia, Y.~Pirola, M.~Previtali, and
  R.~Rizzi.
\newblock Divide and conquer computation of the multi-string {BWT} and {LCP}
  array.
\newblock In {\em CiE}, LNCS, pages 107--117. Springer, 2018.

\bibitem{burrows1994block}
M.~Burrows and D.J. Wheeler.
\newblock {A Block Sorting data Compression Algorithm}.
\newblock Technical report, DEC Systems Research Center, 1994.

\bibitem{claude2015wavelet}
F.~Claude, G.~Navarro, and A.~Ord{\'o}nez.
\newblock {The wavelet matrix: An efficient wavelet tree for large alphabets}.
\newblock {\em Information Systems}, 47:15--32, 2015.

\bibitem{CGRS_JDA_2016}
A.J. Cox, F.~Garofalo, G.~Rosone, and M.~Sciortino.
\newblock Lightweight {LCP} construction for very large collections of strings.
\newblock {\em J. Discrete Algorithms}, 37:17--33, 2016.

\bibitem{EgidiAMB2019}
L.~Egidi, F.A. Louza, G.~Manzini, and G.P. Telles.
\newblock {External memory BWT and LCP computation for sequence collections
  with applications}.
\newblock {\em Algorithms Mol. Biol.}, 14(1):6, 2019.

\bibitem{egidi2017lightweight}
L.~Egidi and G.~Manzini.
\newblock {Lightweight BWT and LCP merging via the Gap algorithm}.
\newblock In {\em SPIRE}, LNCS, pages 176--190. Springer, 2017.

\bibitem{ferragina2000opportunistic}
P.~Ferragina and G.~Manzini.
\newblock Opportunistic data structures with applications.
\newblock In {\em FOCS}, pages 390--398. IEEE, 2000.

\bibitem{8712716}
J.~{Fuentes-Sepúlveda}, G.~{Navarro}, and Y.~{Nekrich}.
\newblock Space-efficient computation of the burrows-wheeler transform.
\newblock In {\em 2019 Data Compression Conference (DCC)}, pages 132--141,
  2019.

\bibitem{navarro2012wavelet}
N.~Gonzalo.
\newblock Wavelet trees for all.
\newblock {\em J. Discrete Algorithms}, 25:2 -- 20, 2014.

\bibitem{Grossi:2005:CSA:1093654.1096192}
R.~Grossi and J.~S. Vitter.
\newblock Compressed suffix arrays and suffix trees with applications to text
  indexing and string matching.
\newblock {\em SIAM J. Comput.}, 35(2):378--407, 2005.

\bibitem{GuerriniRosoneAlcob2019}
V.~Guerrini and G.~Rosone.
\newblock {Lightweight Metagenomic Classification via eBWT}.
\newblock In {\em Algorithms for Computational Biology}, volume 11488 LNBI,
  pages 112--124. Springer International Publishing, 2019.

\bibitem{holt2014constructing}
J.~Holt and L.~McMillan.
\newblock {Constructing Burrows-Wheeler transforms of large string collections
  via merging}.
\newblock In {\em ACM-BCB}, pages 464--471. ACM, 2014.

\bibitem{holt2014merging}
J.~Holt and L.~McMillan.
\newblock {Merging of multi-string BWTs with applications}.
\newblock {\em Bioinformatics}, 30(24):3524--3531, 2014.

\bibitem{Hon:2009:BTB:1654348.1654351}
W.-K. Hon, K.~Sadakane, and W.-K. Sung.
\newblock Breaking a time-and-space barrier in constructing full-text indices.
\newblock {\em SIAM J. Comput.}, 38(6):2162--2178, 2009.

\bibitem{Karkkainen:2007:FBS:1314704.1314852}
J.~K\"{a}rkk\"{a}inen.
\newblock Fast {BWT} in small space by blockwise suffix sorting.
\newblock {\em Theor. Comput. Sci.}, 387(3):249--257, 2007.

\bibitem{10.1007/978-3-642-02441-2_17}
J.~K{\"a}rkk{\"a}inen, G.~Manzini, and S.~J. Puglisi.
\newblock Permuted longest-common-prefix array.
\newblock In {\em Combinatorial Pattern Matching}, pages 181--192, Berlin,
  Heidelberg, 2009. Springer Berlin Heidelberg.

\bibitem{10.1007/3-540-48194-X_17}
T.~Kasai, G.~Lee, H.~Arimura, S.~Arikawa, and K.~Park.
\newblock Linear-time longest-common-prefix computation in suffix arrays and
  its applications.
\newblock In {\em Combinatorial Pattern Matching}, pages 181--192, Berlin,
  Heidelberg, 2001. Springer Berlin Heidelberg.

\bibitem{Kempa:2019:SSS:3313276.3316368}
Dominik Kempa and Tomasz Kociumaka.
\newblock String synchronizing sets: Sublinear-time bwt construction and
  optimal lce data structure.
\newblock In {\em Proceedings of the 51st Annual ACM SIGACT Symposium on Theory
  of Computing}, STOC 2019, pages 756--767, New York, NY, USA, 2019. ACM.

\bibitem{Louza2017}
F.A. Louza, G.P. Telles, S.~Hoffmann, and C.D.A. Ciferri.
\newblock Generalized enhanced suffix array construction in external memory.
\newblock {\em Algorithms Mol. Biol.}, 12(1):26, 2017.

\bibitem{ManberMyers1993}
U.~Manber and G.~Myers.
\newblock Suffix arrays: A new method for on-line string searches.
\newblock {\em SIAM Journal on Computing}, 22(5):935--948, 1993.

\bibitem{MantaciRRS07}
S.~Mantaci, A.~Restivo, G.~Rosone, and M.~Sciortino.
\newblock An extension of the {B}urrows-{W}heeler {T}ransform.
\newblock {\em Theor. Comput. Sci.}, 387(3):298--312, 2007.

\bibitem{10.1007/978-3-540-27810-8_32}
Giovanni Manzini.
\newblock Two space saving tricks for linear time lcp array computation.
\newblock In Torben Hagerup and Jyrki Katajainen, editors, {\em Algorithm
  Theory - SWAT 2004}, pages 372--383, Berlin, Heidelberg, 2004. Springer
  Berlin Heidelberg.

\bibitem{munro2017space}
J.~I. Munro, G.~Navarro, and Y.~Nekrich.
\newblock Space-efficient construction of compressed indexes in deterministic
  linear time.
\newblock In {\em SODA}, pages 408--424. SIAM, 2017.

\bibitem{Munro:1997:SRB:795663.796328}
J.~I. Munro and V.~Raman.
\newblock Succinct representation of balanced parentheses, static trees and
  planar graphs.
\newblock In {\em Proceedings of the 38th Annual Symposium on Foundations of
  Computer Science}, FOCS '97, pages 118--, Washington, DC, USA, 1997. IEEE
  Computer Society.

\bibitem{Navarro:2016:CDS:3092586}
Gonzalo Navarro.
\newblock {\em Compact Data Structures: A Practical Approach}.
\newblock Cambridge University Press, New York, NY, USA, 1st edition, 2016.

\bibitem{navarro2014optimal}
Gonzalo Navarro and Yakov Nekrich.
\newblock {Optimal dynamic sequence representations}.
\newblock {\em SIAM Journal on Computing}, 43(5):1781--1806, 2014.

\bibitem{10.1007/978-3-642-16321-0_34}
E.~Ohlebusch, J.~Fischer, and S.~Gog.
\newblock Cst++.
\newblock In {\em String Processing and Information Retrieval}, pages 322--333,
  Berlin, Heidelberg, 2010. Springer Berlin Heidelberg.

\bibitem{10.1007/978-3-319-15579-1_46}
A.~Policriti, N.~Gigante, and N.~Prezza.
\newblock {Average Linear Time and Compressed Space Construction of the
  Burrows-Wheeler Transform}.
\newblock In {\em Language and Automata Theory and Applications}, pages
  587--598, Cham, 2015. Springer International Publishing.

\bibitem{prezza2018detecting}
N.~Prezza, N.~Pisanti, M.~Sciortino, and G.~Rosone.
\newblock {Detecting Mutations by eBWT}.
\newblock In {\em WABI 2018}, volume 113 of {\em LIPIcs}, pages 3:1--3:15,
  2018.

\bibitem{prezza2019}
N.~Prezza, N.~Pisanti, M.~Sciortino, and G.~Rosone.
\newblock {{SNPs detection by eBWT positional clustering}}.
\newblock {\em Algorithms Mol. Biol.}, 14(1):3, 2019.

\bibitem{prezza_et_al:LIPIcs:2019:10478}
N.~Prezza and G.~Rosone.
\newblock {Space-Efficient Computation of the LCP Array from the
  Burrows-Wheeler Transform}.
\newblock In {\em 30th Annual Symposium on Combinatorial Pattern Matching (CPM
  2019)}, volume 128 of {\em Leibniz International Proceedings in Informatics
  (LIPIcs)}, pages 7:1--7:18, Dagstuhl, Germany, 2019. Schloss
  Dagstuhl--Leibniz-Zentrum fuer Informatik.

\bibitem{PuglisiTurpin2008}
S.J. Puglisi and A.~Turpin.
\newblock Space-time tradeoffs for longest-common-prefix array computation.
\newblock In {\em ISAAC}, volume 5369 of {\em LNCS}, pages 124--135. Springer,
  2008.

\bibitem{Sadakane:2002:SRL:545381.545410}
K.~Sadakane.
\newblock Succinct representations of lcp information and improvements in the
  compressed suffix arrays.
\newblock In {\em Proceedings of the Thirteenth Annual ACM-SIAM Symposium on
  Discrete Algorithms}, SODA '02, pages 225--232, Philadelphia, PA, USA, 2002.
  Society for Industrial and Applied Mathematics.

\bibitem{Sadakane:2007:CST:1326296.1326297}
K.~Sadakane.
\newblock Compressed suffix trees with full functionality.
\newblock {\em Theor. Comp. Sys.}, 41(4):589--607, 2007.

\bibitem{Shi:1996}
F.~Shi.
\newblock Suffix arrays for multiple strings: A method for on-line multiple
  string searches.
\newblock In {\em ASIAN}, volume 1179 of {\em LNCS}, pages 11--22. Springer,
  1996.

\bibitem{TUSTUMI2016}
William~H.A. Tustumi, Simon Gog, Guilherme~P. Telles, and Felipe~A. Louza.
\newblock An improved algorithm for the all-pairs suffix–prefix problem.
\newblock {\em Journal of Discrete Algorithms}, 37:34 -- 43, 2016.

\bibitem{Valimaki:2010:ECS:1498698.1594228}
N.~V\"{a}lim\"{a}ki, V.~M\"{a}kinen, W.~Gerlach, and K.~Dixit.
\newblock Engineering a compressed suffix tree implementation.
\newblock {\em J. Exp. Algorithmics}, 14:2:4.2--2:4.23, 2010.

\end{thebibliography}


\end{document}